\documentclass[10pt,conference,letterpaper]{IEEEtran}
\usepackage{times,amsmath,epsfig}
\usepackage{url}
\usepackage{multirow}
\usepackage{subfig}
\usepackage{graphicx}
\usepackage{xspace}
\usepackage[T1]{fontenc}
\usepackage{amssymb} % assumes amsmath package installed
\usepackage{balance}
\usepackage[linesnumbered,ruled,vlined]{algorithm2e}

\newtheorem{definition}{Definition}
\newtheorem{example}{Example}
\newtheorem{lemma}{\textsc{Lemma}{\bfseries}{\itshape}}

\newtheorem{theorem}{Theorem}

\newtheorem{corollary}{Corollary}

\newtheorem{observation}{Observation}

\newcommand{\ie}{\emph{i.e.,}\xspace}
\newcommand{\eg}{\emph{e.g.,}\xspace}

\newcommand{\etal}{\emph{et al.}\xspace}

\newcommand{\eat}[1]{}

\title{Time is What Prevents Everything from Happening at Once: Propagation Time-conscious Influence Maximization}
\author{%
{Hui Li{\small $~^{\#1}$}, Sourav S Bhowmick{\small $~^{*2}$}, Jiangtao Cui{\small $~^{\dag3}$}, Jianfeng Ma{\small $~^{\#4}$} }%
\vspace{1.6mm}\\
\fontsize{10}{10}\selectfont\itshape
$^{\#}$\, School of Cyber Engineering, Xidian University, China\\
\fontsize{9}{9}\selectfont\ttfamily\upshape
$^{1}$\,hli@xidian.edu.cn, $^{4}$\,jfma@mail.xidian.edu.cn%
\vspace{1.2mm}\\
\fontsize{10}{10}\selectfont\rmfamily\itshape
$^{*}$\,School of Computer Science and Engineering, Nanyang Technological University, Singapore\\
\fontsize{9}{9}\selectfont\ttfamily\upshape
$^{2}$\,assourav@ntu.edu.sg
\vspace{1.2mm}\\
\fontsize{10}{10}\selectfont\rmfamily\itshape
$^{\dag}$\,School of Computer Science and Technology, Xidian University, China\\
\fontsize{9}{9}\selectfont\ttfamily\upshape
$^{3}$\,cuijt@xidian.edu.cn
}
\begin{document}
\maketitle
\begin{abstract}
The influence maximization (\textsc{im}) problem as defined in the seminal paper by Kempe~\etal  has received widespread attention from various research communities, leading to the design of a wide variety of solutions.\eat{ Effective solutions to the \textsc{im} problem is important to companies for designing recommendation systems and viral marketing strategies.} Unfortunately, this classical \textsc{im} problem ignores the fact that time taken for influence propagation to reach the largest scope can be significant in real-world social networks,  during which the underlying network itself may have evolved. This phenomenon may have considerable adverse impact on the quality of selected seeds and as a result all existing techniques that use this classical definition as their building block generate seeds with suboptimal influence spread.

In this paper, we revisit the classical \textsc{im} problem and propose a more realistic version called \textsc{proteus-im} (\textbf{Pro}pagation \textbf{T}im\textbf{e}-conscio\textbf{us} \textbf{I}nfluence \textbf{M}aximization) to replace it by addressing the aforementioned limitation. Specifically, as influence propagation may take time, we assume that the underlying social network may evolve during influence propagation. Consequently, \textsc{proteus-im} aims to select seeds in the \textit{current network} to maximize influence spread in the future instance of the network at the end of influence propagation process without assuming complete topological knowledge of the \textit{future network}. We propose a greedy and a \textit{Reverse Reachable} (\textsc{rr}) \textit{set}-based algorithms called \textsc{proteus-genie} and \textsc{proteus-seer}, respectively, to address this problem. Our algorithms utilize the state-of-the-art \textit{Forest Fire Model} for modeling network evolution during influence propagation to find superior quality seeds. Experimental study on real and synthetic social networks shows that our proposed algorithms consistently outperform state-of-the-art classical \textsc{im} algorithms with respect to seed set quality.
\end{abstract}

% NOTE keywords are not used for conference papers so do not populate them
% \begin{keywords}
% keyword-1, keyword-2, keyword-3
% \end{keywords}
%
\vspace{0ex}\section{Introduction}\label{sec:intro}
With the emergence of large-scale online social networking applications in the last decade, \textit{influence maximization} in online social networks has been widely considered as one of the fundamental and popular problems in social data management and analytics. In the seminal paper by Kempe \etal~\cite{Kempe2003}, this problem is defined as follows. Given a social network $G$ as well as an \textit{influence propagation} (or \textit{cascade}) model, the problem of \textit{influence maximization} (\textsc{im}) is to find a set of initial users of size $k$ (referred to as \textit{seeds}) so that they eventually influence the largest number of individuals (referred to as \textit{influence spread}) in $G$. Effective solutions to the \textsc{im} problem open up opportunities for commercial companies to design intelligent recommendation systems and viral marketing strategies~\cite{Shirazipourazad:2012:IPA:2396761.2396837}.

Kempe \etal~\cite{Kempe2003} proved that the \textsc{im} problem is NP-hard, and presented an elegant greedy approximate algorithm applicable to several popular cascade models, including the \textit{independent cascade} (\textsc{ic}) model, and etc. A key strength of this algorithm lies in its guarantee that the influence spread is within $(1-1/e)$ of the optimal influence spread where $e$ is the base of the natural logarithm. Since then a large body of work (\eg~\cite{effMax_KDD09_Chen,GLL11,DBLP:conf/sigmod/NguyenTD16,DBLP:journals/pvldb/HuangWBXL17}) have been proposed to improve the efficiency of \textsc{im} techniques as well as quality of influence spread. Variants of this classical \textsc{im} problem have also been proposed in recent times such as topic-aware \textsc{im}~\cite{DBLP:journals/pvldb/ChenFLFTT15}, conformity-aware \textsc{im}~\cite{Li:2013:CCG:2452376.2452415}, and competitive \textsc{im}~\cite{Carnes:2007:MIC:1282100.1282167}. In a latest research, ~\cite{DBLP:conf/sigmod/AroraGR17} provides a uniform benchmark to evaluate these classical \textsc{im} solutions. In summary, this elegant work by Kempe \etal has had significant influence on the research community\footnote{ This work has garnered over 4,800 citations in \textit{Google Scholar} and received the test-of-time award in \texttt{ACM KDD 2014.}}.

%==========================================================================================================
\vspace{0ex}\subsection{Limitations of the Definition of Classical IM Problem} \label{ssec:def}
%=========================================================================================================
The classical \textsc{im} problem and its solution in~\cite{Kempe2003} are grounded on the following implicit assumption.  Assume that it takes $t$ time for influence spread from seed set $S$ to reach the largest number of nodes in a social network $G$. Then, $t$ is assumed to be small so that the topology of $G$ can be assumed to remain static during $t$. Consequently, the topology of $G$ is completely known during the propagation process. This is important as the dynamics of influence propagation for all cascade models in the classical \textsc{im} problem demands that neighbors of a node $v$ are known. For example, consider the popular \textsc{ic} model. In this model, we start with an initial set of active nodes, and the influence propagation unfolds in discrete steps where at step $i$ when a node $v$ becomes active, it gets a single chance to activate each of its inactive neighbor $w$ with probability $p$. If $v$ is successful in activating $w$, then $w$ will become active in step $i+1$. This process continues until no more activations are possible. Clearly, successful realization of this propagation process requires that the neighbors of each node are known so that influence can be propagated to its active neighbor(s).

A large volume of subsequent work on influence maximization (\eg~\cite{effMax_KDD09_Chen,GLL11,DBLP:journals/pvldb/HuangWBXL17,Li:2013:CCG:2452376.2452415}) also implicitly or explicitly make the above assumption as they are built on top of the classical \textsc{im} problem~\cite{Kempe2003}. Unfortunately, recent studies reveal that the aforementioned assumption may not hold in practice as time taken for influence propagation is significant, during which the topology of these networks evolves rapidly. For instance, ~\cite{PhysRevLett.103.038702} tested the spread of web advertisements through emails and websites and justified that on average it takes 1.5 days for an intermediary node to propagate the messages and the spread will not reach the largest scope until at least 8 propagations. That is, each cascade of web advertisement may consume up to two or more weeks. Meanwhile, it has been reported that active users of \textit{Facebook} increased from just a million in 2004 to 1 billion in 2012, 8.57\% growth per month on average~\cite{facebook1b}. Similarly, the number of active users in \textit{Twitter} increased from 100 million in September 2011 to 200 million in December 2012, 4.73\% growth per month~\cite{twitter2m}. In particular, it has been shown in~\cite{DBLP:conf/kdd/LeskovecBKT08} that the number of nodes for \emph{Answers}, \emph{Delicious} and \emph{LinkedIn} grows quadratically to the elapsed weeks; and for \emph{Flickr}, it grows exponentially.  In summary, the above studies show that \textit{influence propagation can take significant amount of time to reach the largest scope (several weeks) during which social networks evolve. }

Due to the aforementioned mismatch between the characteristics of real-world social networks and assumption made by the classical \textsc{im} problem, the quality of seeds selected by a state-of-the-art \textsc{im} algorithm is adversely impacted. In particular, the seeds $S$ selected from $G$ may not maximize the spread of influence due to the evolutionary nature of $G$ during influence propagation process. To elaborate further, suppose the influence propagation of $S$ takes $t^\prime$ time and terminates when there is no other node that can be activated. Meanwhile, the social network $G=(V,E)$ at time point $t$ evolves to $G^\prime=(V^\prime,E^\prime)$ during $t^\prime$. Note that the classical \textsc{im} problem aims to compute $S$ from $G$ at $t$ ignoring its evolution during $t^\prime$. Importantly, $S$ may not exhibit the maximal expected influence in $G^\prime$. That is, the seed set $S^\prime$ in $G^\prime$ may not necessarily be identical to $S$. Note that it is not possible in practice to run a state-of-the-art \textsc{im} algorithm on $G^\prime$ at time $t$ to get $S^\prime$ and then select $k$ seeds from $S^\prime \cap V$ as $S$, which is the ``best'' seed set in $G$ that exhibits the maximal influence in $G^\prime$. \textit{This is because it is unrealistic to assume that the complete topology of $G^\prime$ is known at time $t$}. We further illustrate this problem with the following example.

%=============================================================================================
\begin{figure}[t]
\centering
\epsfig{file=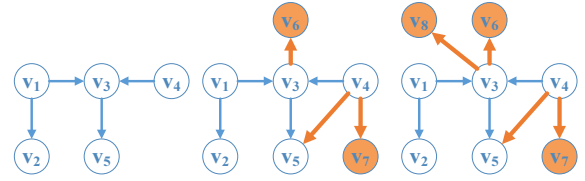, width=0.85\linewidth}
\vspace{0ex}\caption{(a) The original network $G$ at time $t$. (b) The network $G^\prime$ at time $t^\prime$. (c) The network $G^{\prime\prime}$ at time $t^{\prime\prime}$.}\label{fig:fig1}
\hrule
\vspace{0ex}\end{figure}
%=============================================================================================
\begin{example}
Suppose we wish to select one seed ($k=1$) at time $t$ on the social network $G$ depicted in Fig.~\ref{fig:fig1}(a). A state-of-the-art \textsc{im} algorithm will run on $G$ to select $S=\{v_1\}$ as the seed. As influence propagation may take $t^\prime$ time, assume that during this period $G$ evolves to $G^\prime$ as shown in Fig.~\ref{fig:fig1}(b). Specifically, two new nodes (\ie $v_6$ and $v_7$) and three edges (\ie $\overrightarrow{v_4v_5}$, $\overrightarrow{v_4v_7}$ and $\overrightarrow{v_3v_6}$) are added during this time period.  Consequently, $v_1$ may not influence the most number of nodes after the completion of influence propagation as during this period the topology of $G$ has evolved to $G^\prime$. In fact, $v_4$ is a better choice as it may influence more nodes than $v_1$ after $t^\prime$.

At first glance, it may seem that this problem can be easily addressed by running the \textsc{im} technique on $G^\prime$ instead of $G$, which may result in seed set $S=\{v_4\}$. Unfortunately, \textit{it is difficult at time $t$ to predict the topology of $G$ after time $t^\prime > t$ (\ie $G^\prime$) in order to run an \textsc{im} technique on the latter!} In other words, the complete topology of the network at time $t^\prime$ is unknown at time $t$. Observe that this problem occurs regardless of \textit{when} the \textsc{im} algorithm is run. For instance, suppose it is run at time $t^\prime$ on $G^\prime$ to select $S=\{v_4\}$. However, now $G^\prime$ may have evolved to $G^{\prime\prime}$ (Fig.~\ref{fig:fig1}(c)) during the influence propagation process and as a result $v_4$ may not be the optimal seed for maximizing influence anymore.

Fundamentally, the manifestation of this problem is due to the maximization of influence on a network instance $G$ at time $t$ (which is consistent with the classical \textsc{im} problem  definition) instead of discovering seeds that maximize influence on a \textit{future instance} of the network (\ie $G^\prime$ at time $t^\prime$) assuming that influence propagation takes $t^\prime > t$ time.  However, as remarked earlier it is difficult to know the exact topology of the future network $G^\prime$ at time $t$.
\end{example}

%==========================================================================================================
\vspace{0ex}\subsection{Can Recent IM Efforts on Dynamic Networks Address the Limitations?} \label{sec:dyn}\vspace{0ex}
%=========================================================================================================
Recently, several efforts have studied the \textsc{im} problem in the context of dynamic or temporal social networks~\cite{DBLP:conf/icdm/ZhuangSTZS13,DBLP:conf/sdm/ChenSHX15,DBLP:journals/pvldb/OhsakaAYK16,DBLP:journals/pvldb/WangFLT17,DBLP:journals/ton/TongWTD17}. At first glance, it may seem that the aforementioned problem of classical \textsc{im} can be addressed by deploying these techniques as they consider evolutionary nature of the underlying network. \textit{Unfortunately, this is not the case as these techniques either assume that the topology of the network is completely known at a specific time point or are oblivious to the impact of influence propagation time on the network state.} Broadly, these techniques repeatedly run classical \textsc{im} algorithms (or their incremental versions) at different time points in order to find up-to-date seed sets. However, this strategy cannot address the aforementioned limitation of the classical \textsc{im} problem regardless of the way time points are separated or the choice of \textsc{im} algorithm. For instance, suppose $G_i$ at time $t_i$ evolves to $G_j$ at time $t_j$ where $t_j>t_i$. Intuitively, one may select $S_i$ at time $t_i$ in order to maximize the influence at a different temporal state. However, $S_i$ can only assure that the influence is maximized in $G_i$. As remarked earlier, as influence propagates from $S_i$, the network evolves as well. Consequently, whenever influence of $S_i$ reaches the largest scope, the network may have evolved from $G_i$ to $G_j$. Hence, repeatedly selecting seeds using a conventional \textsc{im} algorithm or its variant cannot lead to a superior quality seed set that maximizes influence at a future time point.

As an example, consider the \textit{MaxG} algorithm in~\cite{DBLP:conf/icdm/ZhuangSTZS13}, which ignores the impact of influence propagation time.\eat{ Consequently, although the authors claim that seeds are selected to maximize the influence at any temporal state of the network, the network state when the seeds are selected (\eg $G$ at $t$ in Fig.~\ref{fig:fig1}(a)) is not the one when influence eventually propagates (\eg $G^\prime$ at $t^\prime$ in Fig.~\ref{fig:fig1}(b)).} In other words, it assumes that the topology of the whole network can be easily observed at any timestamp, which is consistent with the assumption made by classical \textsc{im} as discussed earlier. It is worth noting that if the time consumed by the influence propagation process is not ignored, the probability of update operation has to be decayed even if the marginal gain increases. This is because the later a new node is interchanged into $S$, the lesser is the time available for it to propagate to the influence scope as expected. In fact, it has been argued in~\cite{DBLP:conf/cosn/GayraudPT15}, selecting same seeds at different timepoints may result in different influence spread in a dynamic network. \eat{The same problems also persist in~\cite{DBLP:conf/sdm/ChenSHX15}. In summary, as we shall see later, these issues lead to suboptimal seed set selection in~\cite{DBLP:conf/icdm/ZhuangSTZS13}.} The approach in~\cite{DBLP:conf/sdm/AggarwalLY12} firstly separates a time period into several equal-length intervals (of length $h$) based on the \textit{entire} evolution period of the network. Then, an algorithm is presented to select seeds $S$ at $t$ in order to maximize the influence at $t+h$. Unfortunately, it demands the future network state as input in order to compute $h$. For instance, in order to select $S$ in $G$ at time $t$ in Fig.~\ref{fig:fig1}(a), it requires $G^{\prime\prime}$ at time $t^{\prime\prime}$ (Fig.~\ref{fig:fig1}(c)) as input. Obviously, it is unrealistic to assume that $G^{\prime\prime}$ is completely known at $t$.

%=============================================================================================
\vspace{0ex}\subsection{Contributions} \label{ssec:contri}
%=============================================================================================
This paper makes four contributions. First, we theoretically prove that if the aforementioned assumption related to evolutionary nature of a social network and the impact of influence propagation time on seed selection is jettisoned by the classical \textsc{im} problem, then the approximation guarantee for greedy algorithms that the influence spread is within $(1-1/e)$ of the optimal influence spread does not hold anymore. Note that a large number of subsequent work~\cite{effMax_KDD09_Chen,Li:2013:CCG:2452376.2452415,DBLP:conf/cikm/MaYLK08} have used this guarantee as the building block to design new algorithms and derive new results.

Second, we revisit the classical \textsc{im} problem and \textit{redefine} it as \textsc{proteus-im}\footnote{ The name honors \textit{Proteus},  a sea god in Greek mythology, noted for his ability to assume different forms and to prophesy.  The \textsc{proteus-im} problem discovers seeds from a social network that assumes different form from the current instance at the end of the influence propagation process.} (\textbf{Pro}pagation \textbf{T}im\textbf{e}-conscio\textbf{us} \textbf{I}nfluence \textbf{M}aximization) problem by jettisoning the aforementioned assumption made by the former. Intuitively, it is defined as follows. Given a network $G_0=(V_0,E_0)$ at time $t_0$ that may evolve to $G_r=(V_r,E_r)$ at \textit{target time} $t_r$, the goal of the \textsc{proteus-im} problem is to select seeds $S\subseteq V_0$ at time $t_0$ such that information spread from $S$ can reach the largest scope in $G_r$  instead of $G_0$\footnote{ We shall elaborate in Section~\ref{sec:probstat} the justification for choosing seed set from $G_0$ and not solely from $G_r$.}.

Observe that the \textsc{proteus-im} problem differs from the classical \textsc{im} in the following ways. Firstly, we assume that the underlying network evolves during the influence propagation time $t_r$ and the complete topology of the \textit{target network} $G_r$ is unknown at time $t_0$. Secondly, the seeds are selected in a network ($G_0$) whose topology is not identical to the one in which influence finally propagates to the largest scope ($G_r$). In comparison, in the classical \textsc{im} problem these two network topologies are assumed identical. Thirdly, the influence propagation path in our problem may consists of nodes and edges that are currently absent in $G_0$. In comparison, the influence propagation path in classical \textsc{im}, although randomly distributed, strictly sampled from the edges in the current network. We also prove that the \textsc{proteus-im} problem is NP-hard and the expected influence is submodular.

Third, we propose a greedy algorithm called \textsc{proteus-genie} to address the \textsc{proteus-im} problem. Specifically, it selects $k$ nodes at time $t_0$ whose expected influence at time $t_r$ is maximal. A distinguishing feature of the algorithm is that it takes into account evolution of the underlying network during influence propagation process. Note that this is a challenging problem as we cannot make unrealistic assumption that the topology of the network at $t_r$ is completely known \textit{apriori}. To tackle this challenge, we resort to a popular network evolution model called the Forest Fire Model (\textsc{ffm})~\cite{DBLP:conf/kdd/LeskovecKF05} to predict the topology of the network at time $t_r$\footnote{ As our approach is loosely-coupled with the network evolution model, other models can also be adopted in this regard.}. To the best of our knowledge, \textsc{ffm} has never been utilized in the context of \textsc{im}. Specifically, \textsc{proteus-genie} iteratively \textit{selects} nodes with largest marginal gain in expected influence, taking into account the evolution of the network predicted by \textsc{ffm}. The proposed greedy algorithm can be time consuming for large networks as in each iteration we need to simulate network evolution and then select the next optimal seed node. Hence, we propose a Reverse Reachable (\textsc{rr}) set-based algorithm called \textsc{proteus-seer} which significantly reduces the running time while preserving similar influence spread quality. It first selects an \textit{instance number} $\theta$ by utilizing a recent classical \textsc{im} technique~\cite{DBLP:conf/sigmod/TangSX15} and then iteratively predict $\theta$ instances of the target network, $G^\prime_1, \ldots, G^\prime_\theta$. We \textit{select} candidate seeds from each $G^\prime_i$ and aggregate them to finally select the top-$k$ seeds.

Fourth, we investigate the performance of \textsc{proteus-genie} and \textsc{proteus-seer} on real-world social networks. Our experimental study reveals that, as predicted by theory, algorithms designed for the \textsc{proteus-im} problem consistently outperform state-of-the-art classical \textsc{im} techniques in terms of influence spread quality for all datasets, even when the underlying network has changed slightly during influence propagation (\ie $t_r$ may be small). Interestingly, our results emphasize that it is not necessary to possess a complete and accurate knowledge of the topology of $G_r$ to achieve such superior performance. Note that this is important as assuming such complete knowledge renders the \textsc{im} problem unrealistic. Additionally, \textsc{proteus-seer} significantly reduces the running time while preserving similar result quality.

\begin{table}[t]

\centering
\caption{Key notations used in this paper.} \label{T:symbol_definition}
\begin{tabular}{||l|p{63mm}||}
\hline
\textit{Symbol} & \textit{Definition} \\
\hline
$G_i=(V_i,E_i)$ & A social network at time $t_i$ \\
$k$ & The number of seeds \\
$S$ & Seeds set \\
$p$ & Independent cascade probability\\
$\sigma(\cdot)$ & The expected influence function in static network\\
$\sigma_{X_t}(\cdot)$ & The number of influenced node through edges $X_t$\\
$\sigma(\cdot,t)$ & The expected influence of dynamic network at time $t$\\
$\alpha$ & Forward burning probability \\
$\gamma$ & Backward burning ratio \\
$R$ & Number of rounds of simulation in computing expected influence\\
$I$ & Number of rounds for simulating the evolution\\
$\theta$ & Number of predicted instances in \textsc{proteus-seer}\\
$R(v,G)$ & The set of nodes in $G$ that can reach $v$\\
\hline
\end{tabular}
\vspace{0ex}\end{table}

%=============================================================================================
\vspace{0ex}\subsection{Paper Organization}
%=============================================================================================
The rest of this paper is organized as follows. We review classical \textsc{im} techniques in Section~\ref{sec:revistIM}. We formally define the \textsc{proteus-im} problem in Section~\ref{sec:probstat}. Sections~\ref{sec:greedy} and~\ref{sec:heur} present the \textsc{proteus-genie} and \textsc{proteus-seer} algorithms to address this problem. We present the experimental results in Section~\ref{sec:exp}. Finally, we conclude this paper in the last section.
%=============================================================================================
\vspace{0ex}\section{Classical Influence Maximization Problem}\label{sec:revistIM}
%=============================================================================================
In this section, we review related work in classical influence maximization (\textsc{im}) problem for both static and dynamic networks. Table~\ref{T:symbol_definition} describes the key notations used in this paper.
%===========================================================================================================
\vspace{0ex}\subsection{IM in Static Networks}
%===========================================================================================================
Kempe \etal~\cite{Kempe2003} are the first to consider choosing the seeds for \textsc{im} problem as a discrete optimization problem. In their seminal paper, they defined the classical influence maximization problem as follows.

\vspace{0ex}\begin{definition} \label{def:im} \textbf{[Classical Influence Maximization \linebreak Problem]}
\emph{Let $G=(V,E)$ be a network and $\sigma(\cdot)$ be the expected influence of a set of nodes under a given cascade model, measured by the number of nodes that are eventually influenced. Then given a budget $k$, the influence maximization (\textsc{im}) problem aims to select a seed set $S\subset V$ ($|S|=k$) such that the expected influence spread $\sigma(S)$ is maximized, which can be formally described as} $S=\mathop{\arg \max}\limits_{A\subseteq V, |A|=k}\sigma(A).$
\end{definition}

Kempe et al.~\cite{Kempe2003} proposed a general greedy algorithm that returns near optimal results (\ie within $1-1/e$). Since then a large body of \textsc{im} techniques~\cite{Kempe2003,effMax_KDD09_Chen,Li:2013:CCG:2452376.2452415,DBLP:conf/sigmod/TangSX15,irie}  are reported in the literature to improve efficiency, scalability, and influence spread quality. \eat{Furthermore, this classical \textsc{im} problem has been extended to address topic-aware \textsc{im}~\cite{DBLP:journals/pvldb/ChenFLFTT15}, conformity-aware \textsc{im}~\cite{Li:2013:CCG:2452376.2452415}, and influence maximization in competitive networks~\cite{Carnes:2007:MIC:1282100.1282167,Bharathi07competitiveinfluence}.} As highlighted in Section~\ref{sec:intro}, the classical \textsc{im} algorithms assume that the topology of the underlying network is completely known and it does not evolve during influence propagation. Hence, they suffer from the limitations discussed earlier leading to relatively poorer quality of influence spread (detailed in Section~\ref{sec:exp}).

\begin{observation}\label{thm:result}
The seed set $S\subseteq V$ in classical \textsc{im} problem (Definition~\ref{def:im}) may not exhibit the largest expected influence when the underlying network evolves during influence propagation. In other words, $\exists S^\prime\subseteq V$, where $|S^\prime|=|S|=k$, $S^\prime\ne S$, such that $\sigma(S)<\sigma(S^\prime)$.
\end{observation}

\begin{theorem}~\label{thm:grtfl} \emph{The approximation guarantee  that the influence spread is within $(1-1/e)$ of the optimal influence spread for greedy hill-climbing-based classical \textsc{im} algorithms does not hold when the underlying network evolves during influence propagation}.
\end{theorem}
\begin{IEEEproof}
We can easily design a network evolution scenario justifying that $1-1/e$ guarantee does not hold eventually. Without loss of generality, suppose $G_0$ consists of $n$ nodes and only one edge. Let $G_0$ evolve in the following way, at each step, it replicate a copy of $G_0$. For instance, $G_1$ consists of $G_0$ and $G_0^\prime$, which are isomorphic but disconnected with each other; $G_2$ consists of $G_1$ and $G_0^\prime$, and etc. Thus, $G_r$ ($r>k$) will consist of $r+1$ disconnected copies of $G_0$. Suppose we are maximizing the influence spread using greedy hill-climbing-based classical \textsc{im} algorithms in $G_0$. According to the quality guarantee, $\sigma(S_k,r)\ge(1-1/e)\sigma(O_k^0,0)$, and obviously $\sigma(O_k^0,0)=k+1$. Therefore $(1-1/e)(k+1)\le\sigma(S_k,r)\le k+1$. However, we can easily find that $\sigma(O_k^r,r)=2k$, which can be achieved by selecting $1$ seed from each of $k$ disconnected copies of $G_0$ in $G_r$. Obviously, $\forall k\ge 4$, $k+1<(1-1/e)2k$, that is, $\sigma(S_k,r)<(1-1/e)\sigma(O_k^r,r)$.
\end{IEEEproof}

\vspace{0ex}\textbf{Remark.} Kempe \etal showed in~\cite{Kempe2003} that a \textit{non-progressive} (\ie nodes can switch from inactive to active state and vice versa) \textsc{im} problem can be reduced to a \textit{progressive} (\ie nodes can only switch from inactive to active state but not vice versa) case in a different graph. Unfortunately, when the underlying network evolves during influence propagation, the \textsc{im} problem cannot be transformed to a non-progressive case (and subsequently to progressive case) due to the following reasons. They designed a new concept, namely \emph{layered graph}, which is defined as follows. Given a graph $G=(V, E)$ and time limit $\tau$, a layered graph $G^\tau$ on $\tau\dot|V|$ contains a copy $v_t$ for each node $v$ in $G$ and each time step $t \leq \tau$. Firstly, in a non-progressive case, no matter how many layers in the \textit{layered graph} $G^\tau$, the topology of the network $G=(V,E)$ in each layer is completely known. However, in reality a social network does not satisfy this property. That is, if we model the evolving network into a layered graph, then the topology of the network in each layer is unknown and these networks in different layers are different. Secondly, the influence in non-progressive case is measured by the sum over the number of time steps that all nodes $v\in V$ are active. However, in the presence of evolution it cannot be measured in this way. On one hand, $V$ is not fixed in our problem setting; on the other hand, the influence should be measured as the number of active nodes at a target time $t_r$ (\ie the end of propagation in progressive \textsc{im}) instead of summing over the number of steps the nodes are active.

\eat{
\begin{corollary}\label{col:grthd} Given a non-negative, monotone submodular function $\sigma(\cdot,r)$, let $S$ be a set of size $k$ obtained by selecting nodes one at a time, each time choosing the node that exhibits the largest marginal gain in the function. Let $O$ be the ideal set of nodes that maximizes $\sigma(O,r)$ over all possible set with $k$ nodes. Then $\sigma(S,r)\ge(1-1/e)\sigma(O,r)$. That is, $S$ provides a $(1-1/e)$ guarantee for the seed quality.\end{corollary}
\begin{proof}
Reconsider Theorem~\ref{thm:grtfl}. In this scenario $\delta_{i+1}$, the marginal gain between consecutive selection iterations, is now changed to $\max(\sigma(S_i,r)-\sigma(S_{i-1},r))$. Then all the steps in \textbf{Part I} of the proof in Theorem~\ref{thm:grtfl} hold. Consequently, $\sigma(S,r)\ge(1-1/e)\sigma(O,r)$.
\end{proof}
}
\eat{In comparison, we simulate the evolution and predict the future of the network evolution. Based on that, we propose a greedy approach that takes into account the randomly sampled possible future instances of target network and compute the corresponding expected influence for seeds over all the possible instances.}
%==========================================================================================================
\vspace{0ex}\subsection{IM in Dynamic Networks}
%==========================================================================================================
Recently, there have been increasing efforts to address the \textsc{im} problem in dynamic networks. Zhuang \etal~\cite{DBLP:conf/icdm/ZhuangSTZS13} proposed an algorithm called \textit{MaxG} to select seed nodes $S^t$ at a specific time step $t$. It utilizes a heuristic \textit{probing strategy} such that at a target time step, it only needs to probe a limited number of nodes, whose change in the local connections can best uncover the actual influence propagation process. As remarked in Section~\ref{sec:dyn}, it assumes the topology of the whole network can be easily observed at any timepoint. The same limitation also exists in~\cite{DBLP:conf/sdm/ChenSHX15}, which focuses on tracking influential nodes. More recently,~\cite{DBLP:journals/pvldb/OhsakaAYK16} proposed an index model using \textsc{rr} set introduced in~\cite{DBLP:conf/soda/BorgsBCL14} to answer influence maximization query at any temporal state during network evolution. Similar to~\cite{DBLP:conf/icdm/ZhuangSTZS13}, this work also suffers from two key drawbacks. Firstly, it assumes that every atomic evolution step (\eg single node/edge addition) can be fully observed at any timepoint, which is unrealistic in practice. Secondly, it ignores the influence propagation time and fails to anticipate the network state during influence propagation. Consequently, any answer of~\cite{DBLP:journals/pvldb/OhsakaAYK16} (\ie a set of seeds) towards an influence maximization query $q$ at time $t$ may not necessarily generate the expected influence cascade as the network, in which influence eventually propagates, is typically not the same with the one at $t$, based on which it answers $q$.

Aggarwal \etal~\cite{DBLP:conf/sdm/AggarwalLY12} studied the problem of selecting seed nodes $S$ at time $t$, such that a piece of information propagated from these nodes will spread to the largest scope (\ie number of nodes) at time $t+h$, taking into account that the network may evolve during the period from $t$ to $t+h$. However, as discussed in Section~\ref{sec:dyn}, it assumes that the complete topology of the final network where influence eventually propagates to the largest scope is known and seeds are selected from this ``known'' network.
%=============================================================================================
\vspace{0ex}\section{Propagation Time-conscious Influence Maximization}\label{sec:probstat}
%=============================================================================================
\eat{In the preceding sections, we have highlighted limitation of the classical definition of the \textsc{im} problem that may lead to suboptimal influence spread in real-world applications on static and dynamic networks.} In this section, we revisit this decade-old \textsc{im} problem and redefine it to address the aforementioned limitation. We begin by introducing some terminology that we shall be using in this paper. Then, we formally redefine the classical \textsc{im} problem as \textit{propagation time-conscious} \textsc{im} problem.

%=============================================================================================
\vspace{0ex}\subsection{Terminology}
%=============================================================================================
We model a social network as directed graph $G=(V, E)$, where nodes in $V$ represent individuals in the network and edges in $E$ represent relationships between them. The \textit{order} of $G$ is $|V|$ and its \textit{size} is $|E|$. Recall that traditional \textsc{im} assumes influence propagates between nodes according to a specific cascade model and selects $k$ nodes in $V$ as seeds to spread a piece of information such that the information will be propagated to the maximal number of other nodes. However, such influence propagation can take $t_r$ time in reality (which can be several weeks). During this time, the social network may evolve from $G_0=(V_0,E_0)$ at time $t_0$ to $G_r=(V_r, E_r)$ at time $t_r$. We refer to $G_0$ as \textit{current network} and $G_r$ as \textit{target network}. Correspondingly, $t_0$ and $t_r$ are referred to as \textit{current time} and \textit{target time}, respectively. For the sake of generality, we assume that $t_r$ is given by the user as it is application and network dependent. We assume $|V_r| > |V_0|$ and $|E_r| > |E_0|$ as most real-world social networks grow over time. Furthermore, $V_r \cap V_0 \neq \emptyset$  and $E_r \cap E_0 \neq \emptyset$.  We denote the expected influence at time $t$ (\ie the number of influenced nodes at $t$) for seeds $S$ under a given cascade model as $\sigma(S,t)$. For ease of exposition, in the sequel, we assume the independent cascade (\textsc{ic}) model, where influence propagates according to an independent probability $p_{ij}$ along any edge $\overrightarrow{v_iv_j}$, for influence propagation as it is one of the most popular model in the literature. However, our proposed problem is also applicable to other types of cascade models.

%=============================================================================================
\vspace{0ex}\subsection{Redefining IM Problem} \label{ssec:strategy}
%=============================================================================================
The classical influence maximization problem (Definition~\ref{def:im}) ignores the influence propagation time which can be significant in reality, during which the underlying social network may evolve. Hence, we formally redefine this classical influence maximization problem as follows.

\begin{definition} \label{def:dynim} \textbf{[Propagation Time-conscious Influence Maximization Problem]}
\emph{Let $G_0=(V_0,E_0)$ be the current network at time $t_0$ and $k$ be the budget. Let $t_r > t_0$ be the influence propagation time during when $G_0$ evolves to $G_r=(V_r, E_r)$ where $V_r \cap V_0 \neq \emptyset$  and $E_r \cap E_0 \neq \emptyset$. Then, the goal of \textbf{Pro}pagation \textbf{T}im\textbf{e}-conscio\textbf{us} \textbf{I}nfluence \textbf{M}aximization (\textsc{proteus-im}) Problem is to select a set of seed nodes $S \subseteq V_0$ ($|S|=k$) at $t_0$ such that the expected influence spread $\sigma(S,t_r)$ is maximized at $t_r$ assuming that the complete topology of $G_r$ is unknown at $t_0$. That is,} $S=\mathop{\arg \max}\limits_{A\subseteq{V_0}, |A|=k}{\sigma(A,t_r)}.$
\end{definition}

Observe that according to the above definition, seeds are selected from current instance of the network $G_0$ instead of future instances of the network \ie  $G_1, \ldots, G_r$. This is because it is difficult to know at $t_0$ which users may potentially join or leave a social network in the future (before $t_r$), how will they be connected to other users,  and whether they will be part of the seeds.  In fact, as remarked earlier,  it is unrealistic to assume accurate and complete topological knowledge of future instances of the social network (\ie $G_r$) at time $t_0$. Hence, given that influence propagation may take $t_r$ time, it is more realistic to choose a seed set $S\subseteq V_0$ (\ie users who currently exist in the social network) in order to maximize the expected influence spread in the target network $G_r$. Also, observe that in the classical \textsc{im} problem, $G_r = G_0$  as the topology of the network is assumed to be static since $t_r$ is negligible.

\begin{lemma}\label{the:subm}
\emph{The expected influence function at an arbitrary time $t$ for node set $S$ under the \textsc{ic} model, namely $\sigma(\cdot,t)$, defined in Definition~\ref{def:dynim} is sub-modular}.
\end{lemma}

\begin{IEEEproof}
Let $f_{X_t}(v)$ be the set of nodes that can be reached from $v$ on a path comprising of the live edges $X_t$ at time $t$ and $\sigma_{X_t}(A)$ be the number of nodes that can be reached from $A=\{v_0, \ldots, v_k\}$ through $X_t$. In other words, $\sigma_{X_t}(A)=\big|\bigcup_{v\in A}f_{X_t}(v)\big|$. Given two node sets $S\subseteq T$, consider the following expression: $\sigma_{X_t}(S\bigcup\{v\})-\sigma_{X_t}(S)$, which is the number of elements in $f_{X_t}(v)$ that are not already in $\bigcup_{u\in S}f_{X_t}(u)$. It is at least as large as the number of elements in $f_{X_t}(v)$ that are not in $\bigcup_{u\in T}f_{X_t}(u)$. Hence, $\sigma_{X_t}(S\bigcup\{v\})-\sigma_{X_t}(S) \geq \sigma_{X_t}(T\bigcup\{v\})-\sigma_{X_t}(T)$, which is submodular. Moreover, the expected influence of $S$ at $t$ for all possible $X_t$, \ie $\sigma(S,t)$, can be computed as $\sigma(S,t) = \sum_{X_t}Prob[X_t]\cdot\sigma_{X_t}(S).$ According to the equation, $\sigma(S,t)$ can be viewed as a non-negative linear combination of submodular functions, which is also submodular.
\end{IEEEproof}

\begin{lemma}\label{the:nph}
\emph{The \textsc{proteus-im} problem defined in Definition~\ref{def:dynim} under the \textsc{ic} model is NP-hard}.
\end{lemma}

\begin{IEEEproof}
Consider the traditional \textsc{im} problem in Definition~\ref{def:im} which has been proved to be NP-hard~\cite{Kempe2003}. We show that this can be viewed as a special case of the \textsc{new-im} problem.

Given a network $G_0=(V_0,E_0)$ at $t_0$, suppose we are solving \textsc{new-im} over $G_0$ at $t_0$. If $G_0$ remains static for a sufficiently long period until the influence propagation ends at time $t_r$ (\ie the influence reaches the largest scope), in this case $G_0$ is the same with $G_r$ at $t_r$. Therefore, the \textsc{new-im} in $G_0$ is equivalent to \textsc{im} in $G_0$. That is, the problem of maximizing $\sigma(S,t_r)$ in $G_r$ degenerates to the problem of maximizing $\sigma(S)$ in $G_0$. Therefore, the \textsc{new-im} is at least as hard as \textsc{im}, which is NP-hard.
\end{IEEEproof}

Since the \textsc{proteus-im} problem is NP-hard, in the sequel we present two approximate solutions. It is worth emphasizing that given the rich body of work on classical \textsc{im} techniques, our design principle behind these solutions is not to jettison all these efforts but to leverage on the benefits of these techniques wherever possible, while bringing in novel ideas to address the aforementioned limitations of classical \textsc{im}.  Hence, our first solution is a greedy hill-climbing approach called \textsc{proteus-genie}. Our second solution, called \textsc{proteus-seer}, exploits Reverse Reachable (\textsc{rr}) set and is significantly more efficient than \textsc{proteus-genie} while preserving good result quality.

%=============================================================================================
\vspace{0ex}\section{A Greedy Solution}\label{sec:greedy}
%=============================================================================================
\eat{In the preceding section, we have revisited the classical \textsc{im} problem and redefined it to \textsc{proteus-im} problem where given a current network $G_0$ at time $t_0$, the goal is to select $k$ seeds in all $G_i$ ($i\le r$) such that the influence spread at target time $t_r$ is maximized. Fundamentally, in contrast to the classical \textsc{im} problem and its variants, in our problem setting we neither assume that the network is static nor its topology at time $t_r$ is known at time $t_0$.} In this section, we present a novel greedy algorithm called \textsc{proteus-genie} (\textbf{Pro}pagation \textbf{T}im\textbf{e}-conscio\textbf{us} \textbf{G}r\textbf{E}edy selectio\textbf{N} of \textbf{I}nfluential s\textbf{E}eds) that addresses the \textsc{proteus-im} problem.  Observe that designing such algorithm is challenging. On one hand, it is unrealistic to assume complete  knowledge of the topology of the target network $G_r$ at time $t_0$. On the other hand, without knowing the topology of $G_r$ it is very difficult to compute the expected spread in it using existing cascade models (\eg \textsc{ic} model).

We tackle this challenge by \textit{predicting} the expected topology of $G_r$ from $G_0$ by exploiting a popular network evolution model called the \textit{Forest Fire Model}~\cite{DBLP:conf/kdd/LeskovecKF05,DBLP:journals/csur/ChakrabartiF06}.  Consequently, we utilize this predicted topology of $G_r$ to determine the expected spread in it using  an existing cascade model. We begin by briefly introducing this model. Interestingly, as we shall see in Section~\ref{sec:exp}, by leveraging the predicted topology of $G_r$, our proposed algorithms can consistently produce superior quality seeds compared to classical \textsc{im} techniques. That is, we do not need to know the actual topology of $G_r$ at time $t_r$ to produce superior quality seeds!

%=============================================================================================
\vspace{0ex}\subsection{Forest Fire Model (FFM)} \label{ssec:ffmodel}
%=============================================================================================
Majority of social networks are evolutionary in nature and exhibit series of properties and phenomenons, including shrinking diameter, densification power law, etc~\cite{DBLP:conf/kdd/LeskovecKF05}. Several network evolution models~\cite{DBLP:conf/kdd/LeskovecBKT08,DBLP:conf/kdd/LeskovecKF05,DBLP:journals/csur/ChakrabartiF06,DBLP:conf/kdd/BackstromHKL06} have been proposed in the literature to simulate the evolution of real-world online social networks. Among these models, we chose the \textit{Forest Fire Model} (\textsc{ffm})~\cite{DBLP:conf/kdd/LeskovecKF05}, as it outperforms other models~\cite{DBLP:journals/csur/ChakrabartiF06}. Formally, this model is defined as follows.

\begin{definition} \label{def:ffm}
\textbf{[Forest Fire Model]} \emph{Let $G_t$ be a network at time $t$, $G_1$ consist of only the first node. Given an incoming node $v$ at time $t$, the network $G_{t-1}$ at time $t-1$ can be updated to $G_t$ according to the following rules}.
\begin{enumerate} 
  \item \emph{Uniformly select an ambassador node $w$ from $G_{t-1}$ and establish a directed edge from $v$ to $w$}, $\overrightarrow{vw}$.
  \item \emph{Sample two numbers $x$ and $y$, from a pair of binomial distributions whose means are $\alpha/(1-\alpha)$ and $\gamma\alpha/(1-\gamma\alpha)$, respectively. Afterwards, $v$ uniformly selects $x$ in-links and $y$ out-links incident to $w$, respectively. Let $w_1, w_2, \ldots, w_{x+y}$ be the other ends of the selected links. In particular, $\alpha$ is a preset \textbf{forward burning probability}, $\gamma$ is a preset \textbf{backward burning ratio} such that $\gamma\alpha$ is \textbf{backward burning probability}}.
  \item \emph{Establish directed edges from $v$ to $w_1, w_2, \ldots, w_x$, respectively. Similarly, establish directed edges from $w_{x+1}, w_{x+2}, \ldots, w_{x+y}$ to $v$, respectively. Then, we apply step (2) recursively for each of $w_1, w_2, \ldots, w_x$ until there is no new link to be added. As this process continues, nodes can only be visited once such that there is no cyclic sub-structure}.
\end{enumerate}
%\EndOfProof
\end{definition}

It has been shown in~\cite{DBLP:conf/kdd/LeskovecKF05} that the network generated by \textsc{ffm} satisfies majority of real-world network properties, including not only static ones (\eg Heavy-tailed in-degrees and out-degrees~\cite{DBLP:journals/csur/ChakrabartiF06}) but also dynamic ones (\eg Densification Power Law and Shrinking Diameter~\cite{DBLP:conf/kdd/LeskovecKF05}). It has also been demonstrated in~\cite{DBLP:journals/csur/ChakrabartiF06} that evolutions of many real-world networks can be well simulated and predicted using this model. Therefore, we utilize \textsc{ffm} to predict the evolution of a network at target time $t_r$. Specifically, our \textsc{proteus-genie} algorithm \textit{integrates} the \textsc{ffm} with \textit{node selection} during influence maximization to facilitate discovery of superior quality seeds. We now elaborate on this algorithm in detail.

%===========================================================================================================
\vspace{0ex}\subsection{The \textsc{\large proteus-genie} Algorithm} \label{ssec:GffAlg}
%==========================================================================================================
The goal of the \textsc{proteus-genie} algorithm is to greedily select the nodes with the maximal marginal expected influence taking into account the evolution of the underlying network from time $t_0$ to $t_r$ by predicting its topology using \textsc{ffm}.
%It consists of two phases, including seeds initialization and seeds update. The initialization phase is run at $t_0$ before the evolution of $G_0$, aiming to extract the quasi-seeds $S^0\in V_0$ that maximizes the influence spread in $G_r$ taking into account $G_0$ will evolve to $G_r$. Afterwards, as $G_0$ is evolving towards $G_r$, we iteratively updates the selected seeds, such that the influence spread can be maximized in $G_r$ taking into account the eventual evolution logs.

%=============================================================================================
\begin{figure*}[t]
\centering
\epsfig{file=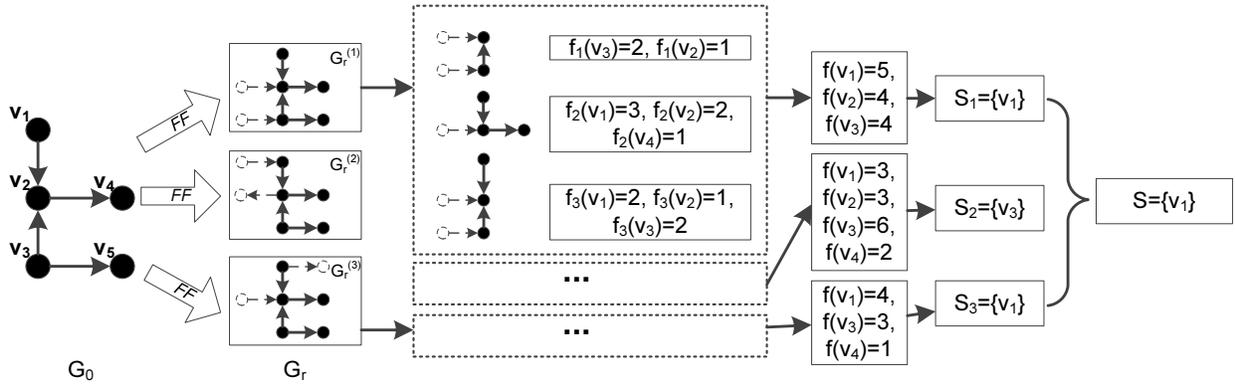, width=0.9\linewidth}
\vspace{0ex}\caption{Greedy seeds selection by the \textsc{proteus-genie} algorithm.}\label{fig:genie}
\hrule
\vspace{0ex} \end{figure*}
%=============================================================================================
Intuitively, seeds selection in \textsc{proteus-genie} is as follows. Firstly, given the current network $G_0$ at time $t_0$, it evaluates the marginal expected influence of all nodes $v\in V_0$  that are predicted to be in $G_{r}$ at time $t_r$, namely $\sigma(v,t_r)$. Note that the topological structure of $G_{r}$ at target time $t_r > t_0$ is generated using \textsc{ffm} based on $G_0$. The forward burning probability and backward burning ratio are selected by fitting the model using the network evolution historical logs before $t_0$. Secondly, it selects the node with the largest expected influence as the first node and removes it from $G_r$. Thirdly, it performs the previous two steps iteratively for $k$ rounds such that it selects $k$ seeds as $S_1$. Observe that in previous steps, we generate one target network $G_r$ using \textsc{ffm}, which results in a deterministic network at time $t_r$. However, the network evolution using \textsc{ffm} during $t_0$ to $t_r$ is a random process which cannot be accurately described using a single-round simulation. Therefore, the previous three steps are executed for $I$ rounds independently, resulting in $I$ different instances of $G_r$, denoted by $G_r^{(1)}, G_r^{(2)}, \ldots, G_r^{(I)}$. Consequently, the seeds sets $S_1, \ldots, S_I$ are generated after $I$ rounds. Finally, it aggregates the ranks of these seeds and selects the top-$k$ nodes with the highest overall ranks as the final seed set $S$. We now formally describe the algorithm.

\begin{algorithm}[t]

\KwIn{Current network $G_0=(V_0,E_0)$, $k$, $p$, $I$, $R$, target time $t_r$, forwards burning probability $\alpha$, backward burning ratio $\gamma$.}
\KwOut{Seed set $S$ of nodes, $|S|=k$}
\Begin{
\ForEach{$i=1$ to $I$}{ %line 2%
    $G_r^{(i)}=(V_r^{(i)}, E_r^{(i)})\leftarrow FF(G_0, \alpha, \gamma, t_r)$\;
    $S_i\leftarrow\emptyset$, $V_0^\prime\leftarrow V_0$\;
    \ForEach{$j=1$ to $k$}{ %line 3%
        \ForEach{$iter=1$ to $R$}{
            generate $G_r^{(i)\prime}=(V_r^{(i)\prime}, E_r^{(i)\prime})$ by removing each edge $\overrightarrow{uv}$ from $G_r^{(i)}=(V_r^{(i)}, E_r^{(i)})$ with probability $1-p$\;
            \ForEach{$v\in V_0$}{
                $f(v) += \sigma_{E_r^{(i)\prime}}(S_i\cup \{v\})-\sigma_{E_r^{(i)\prime}}(S_i)$\;
            }
        }
        $v^*=\mathop{\arg \max}\limits_{v\in V_0}{f(v)}$\;
        $S_i=S_i\cup\{v^*\}$, $S_i(v^*)=j$\;
        $V_0^\prime=V_0^\prime\setminus\{v^*\}$\;
    }
}
$S^*=\bigcup\limits_{i=1}^{I}S_i$\;
\ForEach{$v\in S^*$}{
    $S^*(v)=\sum\limits_{i=1}^{I}{S_i(v)}$\;
}
rank $v\in S^*$ according to $S^*(v)$ in descending order\;
return $S\leftarrow$ top-$k$ items in $S^*$\; %line 10%
}
\caption{The \textsc{proteus-genie} Algorithm.}\label{alg:dgic}
\vspace{0ex}\end{algorithm}

The formal procedure is outlined in Algorithm~\ref{alg:dgic}. Firstly, it simulates the evolution of the network $G_0$ to $G_r$ using \textsc{ffm} (Definition~\ref{def:ffm}) and then initialize a seed set instance $S_i$ as empty (Lines 3-4). Then, it iteratively selects $k$ seed nodes into $S_i$ (Lines 5-12). For the selection of each seed node, we generate graph $G_r^\prime$ by removing each edge in $G_r$ independently with probability $1-p$, resulting in a spanning graph $G_r^\prime=(V_r^\prime,E_r^\prime)$. In this manner, $E_r^\prime$ can be viewed as live edges set $X_r$ at time $t_r$, from which we can compute the marginal influence for each $v\in V_0$. This process repeats for $R$ times and the marginal influences of each node are aggregated (Lines 6-9). Afterwards, the algorithm selects the nodes with the maximal accumulated marginal influence so far (denoted as $v^*$) and inserts it into $S_i$ and removes it from $V_0^\prime$. Meanwhile, it also records the rank of each seed in $S_i$ as $S_i(\cdot)$. The above steps are iteratively performed for $I$ times, until each of $S_1, \ldots, S_I$ is filled with $k$ seeds (Lines 2-12).

So far, we have $I$ instances of seed node set, each of which consists of $k$ nodes as well as their ranks $S_i(\cdot)$. Hence, for each of the nodes $v$ that appears in $S_1, \ldots, S_I$ at least once, the algorithm aggregates its ranks (Lines 13-15). Finally, it selects the top-$k$ nodes as the final seed set $S$ (Lines 16-18).

%Further, we shall illustrate how the seeds can be updated according to the eventual evolution logs of the network. Intuitively, as new nodes and edges eventually join the network, the expected influence of many nodes will definitely be affected. As a result, the seeds selected in the last timestamp may not be so effective any more.

\begin{example}
Consider Fig.~\ref{fig:genie}. Suppose $k=1$, $R = 3$, and $I=3$. Let the current network at time $t_0$ is $G_0$ as shown in the left-hand side of the figure. First, the \textsc{proteus-genie} algorithm utilizes \textsc{ffm} to randomly predict an instance of the target network  $G_r^{(1)}$ at time $t_r > t_0$. Then it randomly removes each edge with probability $1-p$ for $R$ times from $G_r^{(1)}$. This results in three instances of influence. Accordingly, it finds a ranked seed set $S_1$ consisting of the top-$1$ node with the largest expected influence scope over these instances. Afterwards, these steps are repeated twice (\ie $I-1$) by randomly predicting two other instances of the target network, resulting in $G_r^{(2)}$ and $G_r^{(3)}$. Similarly, the algorithm selects another two seed sets, $S_2$ and $S_3$. Finally, it assembles $S_1, \ldots, S_3$ into a bag of nodes $S^* = [v_1, v_3, v_1]$, from which the top-$1$ node $v_1$ with the maximal frequency is returned.
\end{example}

\begin{theorem}\label{the:alg1comp}
\emph{The time complexity of the \textsc{proteus-genie} algorithm (Algorithm~\ref{alg:dgic}) is }$O(I(kR|E_r|+r|E_r|(|V_0|+r))+I)$.
\end{theorem}

\begin{IEEEproof}
The time complexity of \textit{FF} procedure is \linebreak $O(r|E_r|(|V_0|+r))$~\cite{DBLP:conf/kdd/LeskovecKF05}. In Algorithm~\ref{alg:dgic}, in each iteration for $i=1$ to $I$, the time complexity is $O(kR|E_r|)$~\cite{effMax_KDD09_Chen} plus the complexity of \textit{FF} procedure. Thus, filling all $S_i$ for $i=1, \ldots, I$ requires $O(kR|E_r|+r|E_r|(|V_0|+r))$. Therefore, the overall time complexity of Algorithm~\ref{alg:dgic} is $O(I(kR|E_r|+r|E_r|(|V_0|+r))+I)$.
\end{IEEEproof}

\begin{theorem}
\emph{Let $G_0^r$ be the subgraph that joined $G_0$ during $t_0$ to $t_r$ and $\hat{G}_0^r$ be the subgraph generated by \textsc{ffm} from $t_0$ to $t_r$. If $\hat{G}_0^r$ is identical to $G_0^r$ with probability $\eta$, then each $S_i$ corresponding to $\hat{G}_0^r$ generated by the \textsc{proteus-genie} algorithm (Algorithm~\ref{alg:dgic}) is guaranteed to achieve $(1-1/e)$-approximation for \textsc{proteus-im} with probability $\eta$.}
\end{theorem}
\begin{IEEEproof}
Due to the newly joined edges and nodes, for each $v\in V_0$, the expected influence of it at $G_r$, namely $\sigma(v,r)$ can be separated as two parts, one from original graph $G_0$ (i.e., $\sigma(v,0)$), the other from the marginal increase in the expected influence caused by $G_0^r$, denoted by $f(v,0,r)$. Let $\hat{f}(v,0,r)$ be the latter part, then the expected influence of $v$ in predicted graph $G_r$ can be computed as: $\hat{\sigma}(v,r)=\sigma(v,0)+\hat{f}(v,0,r)$. As $\hat{G}_0^r$ is identical with $G_0^r$ with probability $\eta$, then $\hat{f}(v,0,r)=f(v,0,r)$ with probability $\eta$. Therefore, $\hat{\sigma}(v,r)=\sigma(v,r)$ with probability $\eta$.

Moreover, as \textsc{proteus-im} degenerates to classical hill-climb algorithm if graph $G_r$ is static, which is guaranteed to be within $(1-1/e)$ of the optimal, then the seeds $S$ selected by \textsc{proteus-im} in $G_r$ is within $(1-1/e)$ of the optimal. Putting it together, the seeds selected by \textsc{proteus-im} can achieve $(1-1/e)$-approximation with probability $\eta$.
\end{IEEEproof}

\textbf{Remark.} Typically network evolution may be slower than influence propagation.  However, our framework does not demand any correlation between the time steps of the \textsc{ffm} and the influence propagation time. As long as network is evolving and influence propagation takes time, our proposed model and algorithm fit well. In fact, when \textsc{ffm} is extremely slow (\ie network hardly evolves) and influence propagation is extremely fast (\ie $t_r$ is negligible), the \textsc{proteus-im} problem is close to the classical \textsc{im} problem. Particularly, in the unrealistic case when \textsc{ffm} is very slow (\eg each time only one node is added and we have enough time to grasp the topology of network at any temporal state) and $t_r$ is negligible, \textit{MaxG}~\cite{DBLP:conf/icdm/ZhuangSTZS13} works well. In contrast, whenever influence propagation takes time, our solution fits well.

Furthermore, since in our framework we select seeds from $G_0$ for maximizing influence at $G_r$ (as discussed in Section~\ref{ssec:strategy}), we do not need to care about how $G_0$ evolves to $G_r$. This only matters for techniques (\eg \textit{MaxG}~\cite{DBLP:conf/icdm/ZhuangSTZS13}) where seeds are iteratively selected during the evolution.

%============================================================================================
\vspace{0ex}\section{A Reverse Reachable Set-based Solution}\label{sec:heur}
%===========================================================================================
Observe that the time complexity of \textsc{proteus-genie} is highly influenced by $|E_r|$ and $|V_0|$ (Theorem~\ref{the:alg1comp}). These values are large for real-world networks containing millions of nodes and hence the efficiency of the greedy algorithm can be adversely affected when dealing with such networks. In this section, we address this issue by proposing an algorithm called \textsc{proteus-seer} (\textbf{Pro}pagation \textbf{T}im\textbf{e}-conscio\textbf{us} \textbf{SE}ed S\textbf{E}lection using \textbf{R}R set), which leverages the notion of \textit{Reverse Reachable (\textsc{rr}) Set}~\cite{DBLP:conf/soda/BorgsBCL14} in addition to \textsc{ffm} for seed selection. For the sake of completeness, we first briefly introduce the concept of \textsc{rr} set before discussing our algorithm to address the \textsc{proteus-im} problem.

%===========================================================================================================
\vspace{0ex}\subsection{Reverse Reachable Set}
%==========================================================================================================
Let $v$ be a node in $G$ and $g$ be a graph obtained by removing each edge $e$ in $G$ with probability $1-p$. The \textit{reverse reachable} (\textsc{rr}) set~\cite{DBLP:conf/soda/BorgsBCL14} for $v$ in $g$, denoted as $R(v, g)$, is the set of nodes in $g$ that can reach $v$. That is, for each node $u$ in the \textsc{rr} set, there is a directed path from $u$ to $v$ in $g$. For example, consider Fig.~\ref{fig:fig1}(a). The \textsc{rr} set for node $v_5$ in $G$ contains all nodes in $G$ that can reach $v_5$. That is, $R(v_5,G)=\{v_3,v_1,v_4\}$.

Let $\mathcal{G}$ be the distribution of $g$ induced by the randomness in edge removals from $G$. A \textit{random} \textsc{rr} set~\cite{DBLP:conf/soda/BorgsBCL14} is an \textsc{rr} set generated on an instance of $g$ randomly sampled from $\mathcal{G}$ for a node selected uniformly at random from $g$.

Note that the notion of \textsc{rr} set is currently the most efficient and promising way to answer influence maximization problem with guaranteed result quality, and has been recently deployed in~\cite{DBLP:journals/pvldb/OhsakaAYK16,DBLP:conf/soda/BorgsBCL14} to generate ``near-optimal'' solution for the \textsc{im} problem. However, these techniques either assume the network is static or ignore the influence propagation time. More importantly, they cannot be trivially extended to handle the \textsc{proteus-im} problem. \eat{Specifically, in these efforts the \textsc{rr} set is computed based on the assumption that the underlying network is static and its topology is known at any timepoint.} \eat{In contrast, in the \textsc{proteus-im} problem we need to compute the \textsc{rr} set by taking into account the fact that the network may evolve during influence propagation based on randomly predicting the network change using Forest Fire model. Furthermore, as the future network $G_r$ is randomly predicted several times, we utilize a bag of nodes to assemble all instances of the reverse reachable sets computed from different random instances of predicted future network.}

\begin{algorithm}[t]

\KwIn{Current network $G_0=(V_0,E_0)$,  seeds number $k$, influence probability $p$, $\theta$, target time $t_r$, forward burning probability $\alpha$, backward burning ratio $\gamma$.}
%\KwOut{Seed set $S$ of nodes, $|S|=k$}
\KwOut{Seed set $S$ of nodes, $|S|=k$.}
\Begin{
$S \leftarrow \emptyset$\;
\ForEach{$i$ in $1$ to $\theta$}{ %line 2%
    $G_r^{(i)}=(V_r^{(i)}, E_r^{(i)})\leftarrow FF(G_0, \alpha, r, t_r)$ \;
    uniformly sample $v_i\in V_r^{(i)}$\;
    initialize $RR$ set for $v_i$ as: $R(v_i,G_r^{(i)})=\{v_i\}$\;
    \ForEach{$v^\prime\in R(v_i,G_r^{(i)¡¯})$}{
        \ForEach{$\overrightarrow{wv^\prime}\in E_r^{(i)¡¯}$}{
            with probability $p$, let $R(v_i,G_r^{(i)¡¯})=R(v_i,G_r^{(i)¡¯})\cup\{w\}$\;
        }
    }
}
$\mathcal{R} = \{R(v_1,G_r^{(1)¡¯}), R(v_2,G_r^{(2)¡¯}), \ldots, R(v_\theta,G_r^{(\theta)¡¯})\}$\;
$\mathcal{RS} = \bigcup\limits_{R_i\in\mathcal{R}}{R_i}$\;
\ForEach{$i=1$ to $k$}{
    \ForEach{$w\in\mathcal{RS}$}{
        $ct(w,\mathcal{R}) = \sum\limits_{R_j\in\mathcal{R}}{ct(w, R_j)}$\;
    }
    $u = \mathop{\arg \max}\limits_{w\in\mathcal{RS}}{ct(w,\mathcal{R})}$\;
    $S = S\cup\{u\}$\;
    \ForEach{$R_j\in\mathcal{R}$ that contains $u$}{
        remove $R_j$ from $\mathcal{R}$\;
    }
}
return $S$\;
}
\vspace{0ex}\caption{The \textsc{proteus-seer} Algorithm.}\label{alg:her}
\end{algorithm}

%===========================================================================================================
\vspace{0ex}\subsection{The \textsc{\large proteus-seer} Algorithm} \label{ssec:heuristic}
%==========================================================================================================
The key idea of the algorithm is to compute the \textsc{rr} set by considering the evolution of the network due to random prediction using \textsc{ffm}. Since the target network is randomly predicted several times, we utilize a bag of nodes to assemble all instances of the \textsc{rr} sets computed from different random instances of the predicted network. Specifically, our algorithm comprises of the following key steps.

\begin{enumerate} 
  \item First, we use the \textsc{ffm} to simulate the evolution of $G_0$, from which we get $G_r$. This process is iteratively repeated for $\theta$ times, such that we can get $\theta$ different instances of $G_r$, denoted by $\mathcal{G}_r = \{G_r^{(1)}, G_r^{(2)}, \ldots, G_r^{(\theta)}\}$. In particular, $\theta$ is computed as in~\cite{DBLP:conf/sigmod/TangSX15}.
  \item Second, for each instance $G_r^{(i)}$, uniformly sample a node from $V_r$ as $v_i$, and generate a \textsc{rr} set for it denoted as $R(v_i,G_r^{(i)})$. Consequently, we have $\theta$ such sets, each corresponds to a sampled node. In the sequel, we denote these sets as $\mathcal{R}=\{R(v_1, G_r^{(1)}), R(v_2, G_r^{(2)}), \ldots, R(v_\theta, G_r^{(\theta)})\}$
  \item Finally, we greedily select from all \textsc{rr} sets in $\mathcal{R}$ the node $w$ which appears in the most number of \textsc{rr} sets and then remove these sets from $\mathcal{R}$. We iteratively select $k$ such nodes and then output them as final seeds set $S$.
\end{enumerate}

%=============================================================================================
\begin{figure}[t]
\centering
\epsfig{file=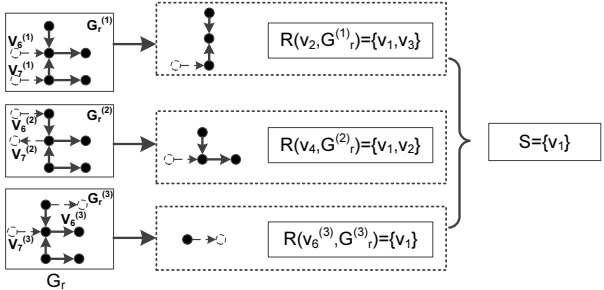, width=0.9\linewidth}
\vspace{0ex}\caption{RR set-based seed selection.}\label{fig:heuristic}
\hrule
\vspace{0ex}\end{figure}
%=============================================================================================

Algorithm~\ref{alg:her} outlines the formal procedure.  Similar to the \textsc{proteus-genie} algorithm, it simulates the evolution of $G_0$ based on \textsc{ffm} to generate an instance of target network $G_r^{(i)} = (V_r^{(i)}, E_r^{(i)})$ (Line 4). Based on $G_r^{(i)}$, it uniformly samples a node $v_i\in V_r^{(i)}$ and generates a random \textsc{rr} set of this node with respect to $G_r^{(i)}$, resulting in $R(v_i, G_r^{(i)})$ (Lines 5-9). The generation of each \textsc{rr} set is implemented as a randomized breath-first search on $G_r^{(i)}$. Given a node $v_i$, it first creates an empty queue and then flips a coin for each incoming edge $e$ of $v_i$. It retrieves the node $u$ with probability $p$ from which $e$ starts and inserts it into the queue. Subsequently, the algorithm iteratively extracts the node $v¡ä$ at the top of the queue and examines each incoming edge $e¡ä$ of $v_i$. If $e¡ä$ starts from an unvisited node $u¡ä$, it adds $u¡ä$ into the queue with probability $p$. This iterative process terminates when the queue becomes empty. Finally, the algorithm collects all nodes visited during this process (including $v_i$) and use them to form $R(v_i, G_r^{(i)})$.

The aforementioned steps are repeated (in parallel) for $\theta$ times resulting in $\mathcal{R} = \{R(v_1, G_r^{(1)^\prime}), \ldots, R(v_\theta, G_r^{(\theta)^\prime})\}$ (Line 10). Let $\mathcal{RS}$ be the set of nodes that appear in any of these $R(v_i, G_r^{(i)^\prime})$ (Line 11). For all the nodes in $\mathcal{RS}$, the algorithm greedily selects the one, say $u$, which appears in the most number of \textsc{rr} sets in $\mathcal{R}$ (Lines 13-16), indicating that $u$ can reach the maximal number of nodes in $v_1, v_2, \ldots, v_\theta$. Then it removes from $\mathcal{S}$ the \textsc{rr} sets where $u$ appears (Lines 17-18). In particular, we denote $ct(u, R(v_i, G_r^{(i)^\prime}))=1$ if $u\in R(v_i, G_r^{(i)^\prime})$, and $0$ otherwise. This seeds selection process is iteratively performed for $k$ rounds to identify the final set of seed nodes (Lines 12-18).

\begin{table}[t]

\vspace{0ex}
\centering \caption{Datasets.}
\label{t:data}
\begin{tabular}{||l|p{11mm}|p{13mm}|p{30mm}||}\hline
\textit{Network} & $\#$\textit{nodes} & $\#$\textit{edges}  & \textit{Degree of Change} (DoC) \\ \hline
\textsf{Syn-$G_1$} & 4,000 & 16,033 & - \\
\textsf{Syn-$G_2$} & 4,500 & 17,309 & Low \\
\textsf{Syn-$G_3$} & 5,000 & 18,512 & Low \\\hline
\textsf{Ph-$G_1$} & 32,354 & 264,963 & -\\
\textsf{Ph-$G_2$} & 38,558 & 347,268 & Medium\\
\textsf{Pa-$G_1$} & 1,061,606 & 1,365,903 & - \\
\textsf{Pa-$G_2$} & 1,772,362 & 5,452,113 & High\\
\textsf{Pa-$G_3$} & 2,436,431 & 11,437,592 & High\\
\textsf{Pa-$G_4$} & 3,774,768 & 16,518,948 & High\\
\hline
\end{tabular}
\vspace{0ex}\end{table}

\eat{\begin{example} {\em
Consider the current network $G_0$ at time $t_0$ in Fig.~\ref{fig:genie}. Suppose $k=1$ and $\theta=3$. The \textsc{proteus-seer} algorithm randomly predicts an instance of the target network $G_r^{(1)}$ as shown in Fig.~\ref{fig:heuristic}. From $G_r^{(1)}$, it uniformly samples a node $v_2\in V$ and generates its \textsc{rr} set by following the steps mentioned above. Let $R(v_2,G_r^{(1)})=\{v_1,v_3\}$.  Afterwards, it repeats these steps for another two times by randomly predicting two other instances of the target network, resulting in $G_r^{(2)}$ and $G_r^{(3)}$. Hence, another two \textsc{rr} sets, namely $R(v_4,G^{(2)}_r)$ and $R(v_6^{(3)},G^{(3)}_r)$, are generated. Finally, these \textsc{rr} sets are assembled into a bag of nodes $S^*$, from which the top-1 node with the maximal frequency, \ie $v_1$, is returned as output.}
\end{example}}

\begin{theorem}\label{the:alg2comp}
\emph{The time complexity of the \textsc{proteus-seer} Algorithm is} $O((|E_r|+|V_r|)\log|V_r|+r|E_r||V_r|+k|V_r|)$.
\end{theorem}

\begin{IEEEproof}
As reported in~\cite{DBLP:conf/sigmod/TangSX15}, the time complexity of computing $\theta$ is $O(\ell(|E_r|+|V_r|)\log|V_r|)$ where $\ell$ is a quality factor which controls the results quality. The time complexity of evolution simulation based on FFM is $O(r|E_r||V_r|)$. The process of generating a \textsc{rr} set requires a \textsc{bfs}, which is $O(|E_r|+|V_r|)$. As $\theta$ has been computed and fixed, then generating all different \textsc{rr} sets requires $O(\ell(|E_r|+|V_r|)\log|V_r|+r|E_r||V_r|)$. The time complexity of seeds selection is $O(k|V_r|)$. Moreover, $\ell$ is a predefined quality parameter (always set as $1$). Therefore, the overall time complexity of Algorithm~\ref{alg:her} is $O((|E_r|+|V_r|)\log|V_r|+r|E_r||V_r|+k|V_r|)$.
\end{IEEEproof}

%=============================================================================================
\vspace{0ex}\section{Experiments}\label{sec:exp}
%=============================================================================================
In this section, we investigate the performance of \textsc{proteus-genie} and \textsc{proteus-seer}. All algorithms considered for our investigation are implemented in C++. We ran all experiments on 3.2GHz Quad-Core Intel i7 machines with 16\textsc{gb} \textsc{ram}, running Windows 7. Note that there is no existing \textsc{im} algorithm that addresses the \textsc{proteus-im} problem. Hence, we are confined to use state-of-the-art algorithms designed for the classical \textsc{im} problem as baseline methods. \eat{Also, we do compare our algorithms with~\cite{DBLP:conf/icdm/ZhuangSTZS13} in the aspect of influence quality although it requires complete knowledge of the topology of the network at any time step.}

Specifically, we investigate the following key issues. (1) Is the seed set selected at target time $t_r$ differs significantly from those selected at current time $t_0$? (2) Do our proposed algorithms designed for the \textsc{proteus-im} problem consistently produce superior quality seeds compared to state-of-the-art algorithms designed for the classical \textsc{im} problem? (3) Is the running time of the \textsc{proteus-seer} algorithm reasonable for large networks without significantly compromising the quality of influence spread? (4) What is the impact of various parameters (\eg $t_r$, $p$, $I$) on the performance of the proposed algorithms?

%======================================================================================================
\vspace{0ex}\subsection{Experimental Setup} \label{sec:setup}
%======================================================================================================
\noindent \textbf{\underline{Datasets}.} Recall from Section~\ref{sec:intro}, influence propagation may take several weeks to months and different networks may evolve at varying rates during this period. Hence, we choose real-world and synthetic datasets for our experiments to represent these varying degree of change (DoC). Table~\ref{t:data} summarizes these datasets. We use two real-world datasets to generate different snapshots of a network representing different degree of evolution. The first one is high-energy physics (Hep) paper citation networks collected through \textit{Arxiv}\footnote{ \url{http://arxiv.com}} during the period from January 1993 to April 2003 (124 months). It contains the historical logs for the appearance timestamp of each paper as well as its citation links\footnote{ Downloaded from \url{http://www.cs.cornell.edu/projects/kddcup/datasets.html}.}. Since each node is associated with a timestamp indicating when it has joined the network, we can construct different instances of the social network at different time points. The networks \textsf{Ph-$G_1$} and \textsf{Ph-$G_2$} represent two temporal states of the citation graph\eat{ at the end of 2000 and 2002, respectively}. The second dataset, \emph{Patents}\footnote{ Downloaded from \url{http://www.nber.org/patents/}.}, comprises of information on almost 3 million U.S. patents granted between January 1963 and December 1999 and all citations made to these patents between 1975 and 1999 (over 16 million). A specific temporal state is extracted by selecting all citations (edges) that appear before a specific timestamp. In particular, \textsf{Pa-$G_1$}, \textsf{Pa-$G_2$}, \textsf{Pa-$G_3$} and \textsf{Pa-$G_4$} are selected as four representative temporal states of this citation network\eat{at 1979, 1989, 1995 and 1999, respectively}. Note that we can extract any temporal states (e.g., weekly, monthly, or yearly) from these two networks. We also generate synthetic datasets using the Forest Fire model\footnote{ According to steps described in \url{http://snap.stanford.edu/snap-1.8/download.html}.} (with default $\alpha = 0.35$ and $\gamma = 0.32$) in order to simulate snapshots of a network with small degree of changes. Specifically, we generate three temporal snapshots with slightly varying number of nodes and edges, denoted by \textsf{Syn-$G_1$}, \textsf{Syn-$G_2$}, and \textsf{Syn-$G_3$}.

The last column in Table~\ref{t:data} specifies the degree of change in the network w.r.t. the number of nodes and edges compared to the preceding snapshot. In summary, the synthetic datasets represent networks with \textit{small} degree of evolution. The real-world datasets, on the other hand, represent networks with \textit{moderate} (\emph{Hep}) or \textit{high} (\emph{Patents}) degree of change. It is worth emphasizing that different real-world networks may have different degree of evolution during influence propagation. Hence, the seed set selection is impacted by the evolution characteristics of the underlying network as well as $t_r$.

\noindent \textbf{\underline{Forest Fire Model (FFM) parameters}.} As discussed in Section~\ref{ssec:GffAlg}, the forward burning probability and backward burning ratio are selected by fitting the model using the network evolution historical logs before $t_0$. That is, we select the states of networks before \textit{Ph-$G_1$} and \textit{Pa-$G_1$}, and fit \textsc{ffm} accordingly. Specifically, we set $\alpha=0.19, \gamma=0.75$ for \textit{Ph-$G_i$}; $\alpha=0.15, \gamma=0.76$ for \textit{Pa-$G_i$}.

\noindent \textbf{\underline{Algorithms}.} We run the following \textsc{im} algorithms under \textsc{ic} model (with $p=0.01$) for our experimental study:

\begin{figure}[t]
\centering
\subfloat[][\textsf{Syn-$G_1 (G_0)$,Syn-$G_3 (G_r)$}\label{f3a}]{\epsfig{file=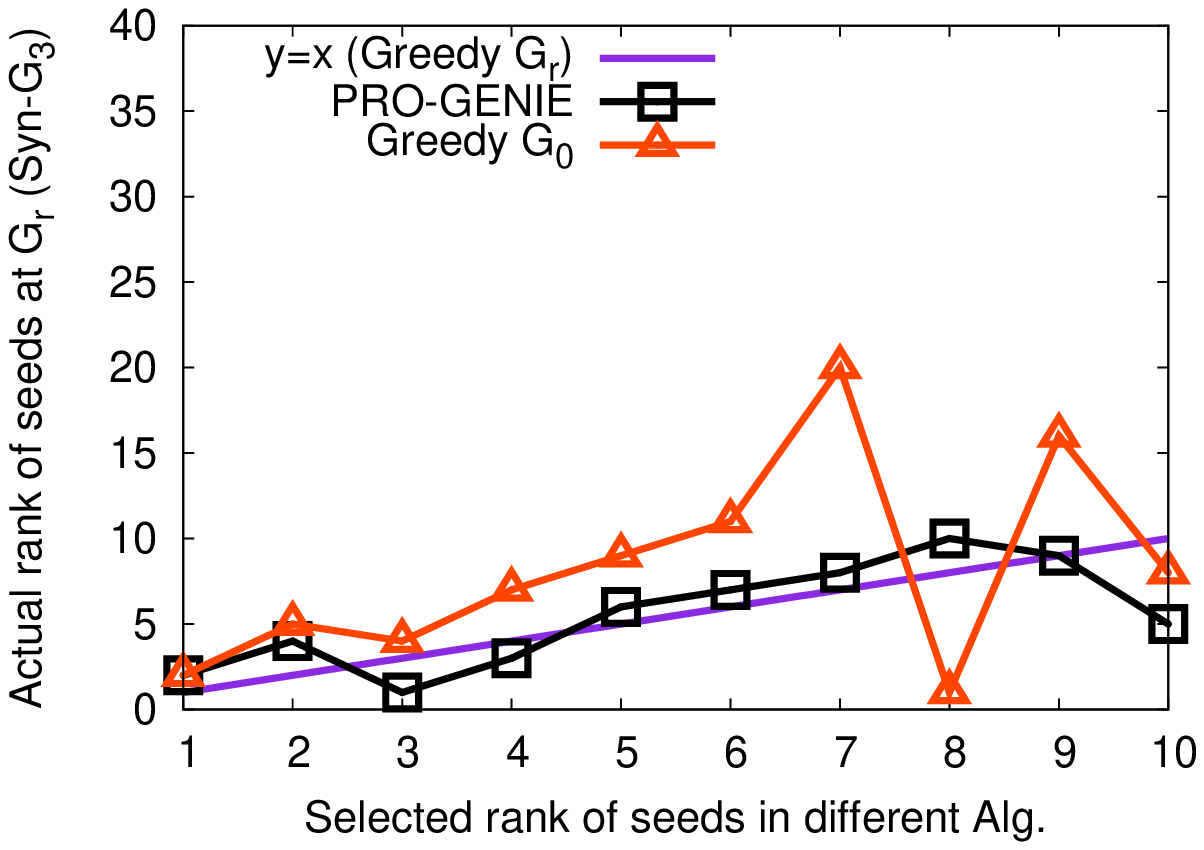, width=0.45\linewidth}}
\hspace{2ex}\subfloat[][\textsf{Syn-$G_2 (G_0)$},\textsf{Syn-$G_3(G_r)$}\label{f3b}]{\epsfig{file=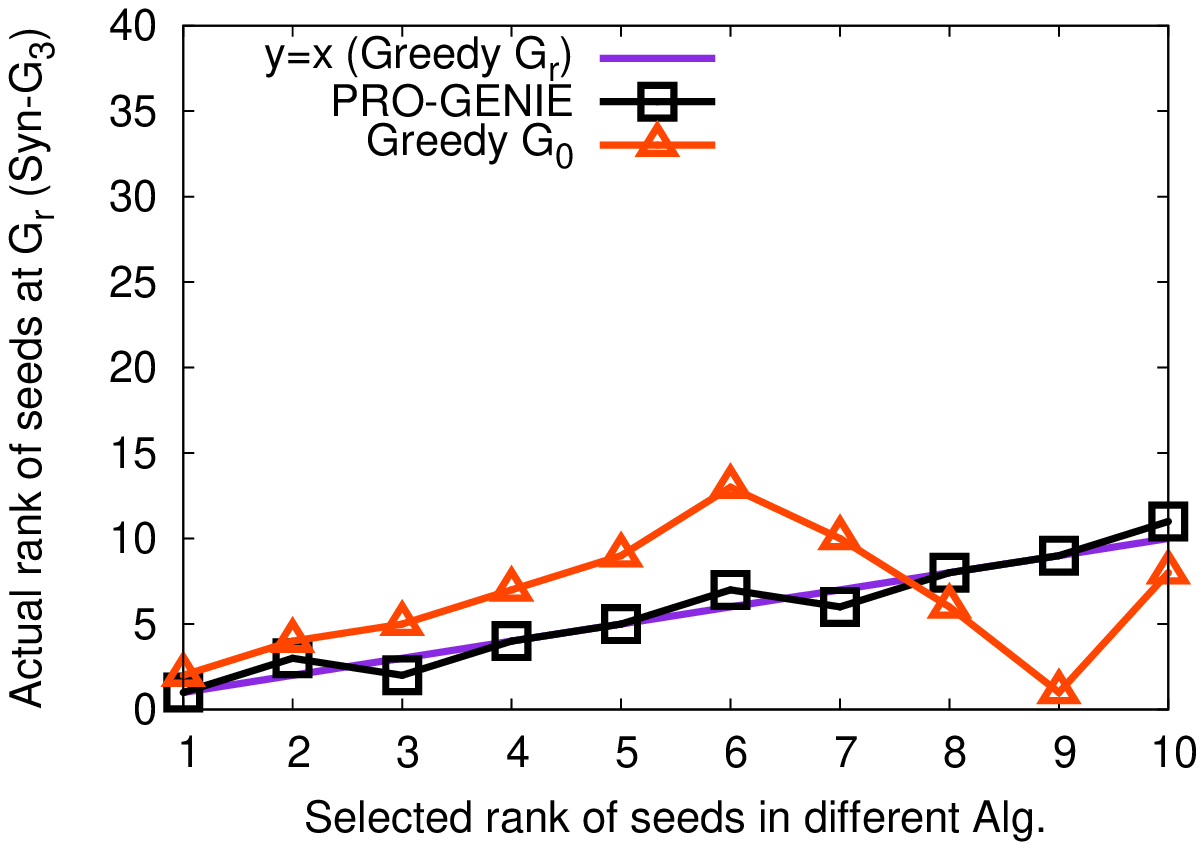, width=0.45\linewidth}}
%\subfloat[][Pa-$G_3 (G_0)$, Pa-$G_4 (G_r)$ \label{f3c}]{\epsfig{file=syn50-100.eps, width=0.23\linewidth}}
%\subfloat[][Ph-$G_1 (G_0)$, Ph-$G_2 (G_r)$ \label{f3d}]{\epsfig{file=fig/real32-38.eps, width=0.23\linewidth}}
\vspace{0ex}\caption{Ranks of seeds generated by different algorithms at current and target times.}\label{fig3}
\hrule
\vspace{0ex}\end{figure}

\begin{itemize} 

\item ``\emph{Greedy}'': \emph{MixGreedyIC} algorithm~\cite{effMax_KDD09_Chen} to address the classical \textsc{im} problem, as it exhibits the best seeds quality.

\item ``\textsc{irie}'': \textsc{irie} algorithm proposed in~\cite{irie}.

\item ``\textsc{imm}'': \textsc{imm} algorithm proposed in~\cite{DBLP:conf/sigmod/TangSX15}.

\item ``\emph{MaxG}'': \emph{MaxG} algorithm (with $\epsilon=0.01$)~\cite{DBLP:conf/icdm/ZhuangSTZS13}, which is a dynamic \textsc{im} algorithm that requires the full knowledge of network evolution.  Specifically, it assumes (a) the complete evolution logs of the network is known; (b) each time a new node arrives, there is sufficient time to update the seeds; and (c) the influence propagation time is negligible.

\item ``\textsc{pro-genie}'': The greedy algorithm in Section~\ref{sec:greedy} for the \textsc{proteus-im} problem.

\item ``\textsc{pro-seer}'': The \textsc{rr} set-based method proposed in Section~\ref{sec:heur} to address the \textsc{proteus-im} problem.

\end{itemize}

Unless specified otherwise, we set $I=500$ and $R=5,000$ for \textsc{pro-genie}, \textsc{pro-seer}, and \emph{Greedy}, respectively. \eat{In the following, we shall present the experimental results in both synthetic and real-world datasets, respectively.}

%======================================================================================================
\vspace{0ex}\subsection{Experimental Results}
%======================================================================================================

\begin{figure*}[t]
\centering
%\subfloat[][\label{f9a}]{\epsfig{file=fiinfsprdg1g2g3.eps, width=0.4\linewidth}}
%\subfloat[][\label{f9b}]{\epsfig{file=fig/infsprdpa1pa2pa3.eps, width=0.4\linewidth}}
\subfloat[][$G_0=$\textsf{Pa-$G_1$} ($k=50$)\label{f9e}]{\epsfig{file=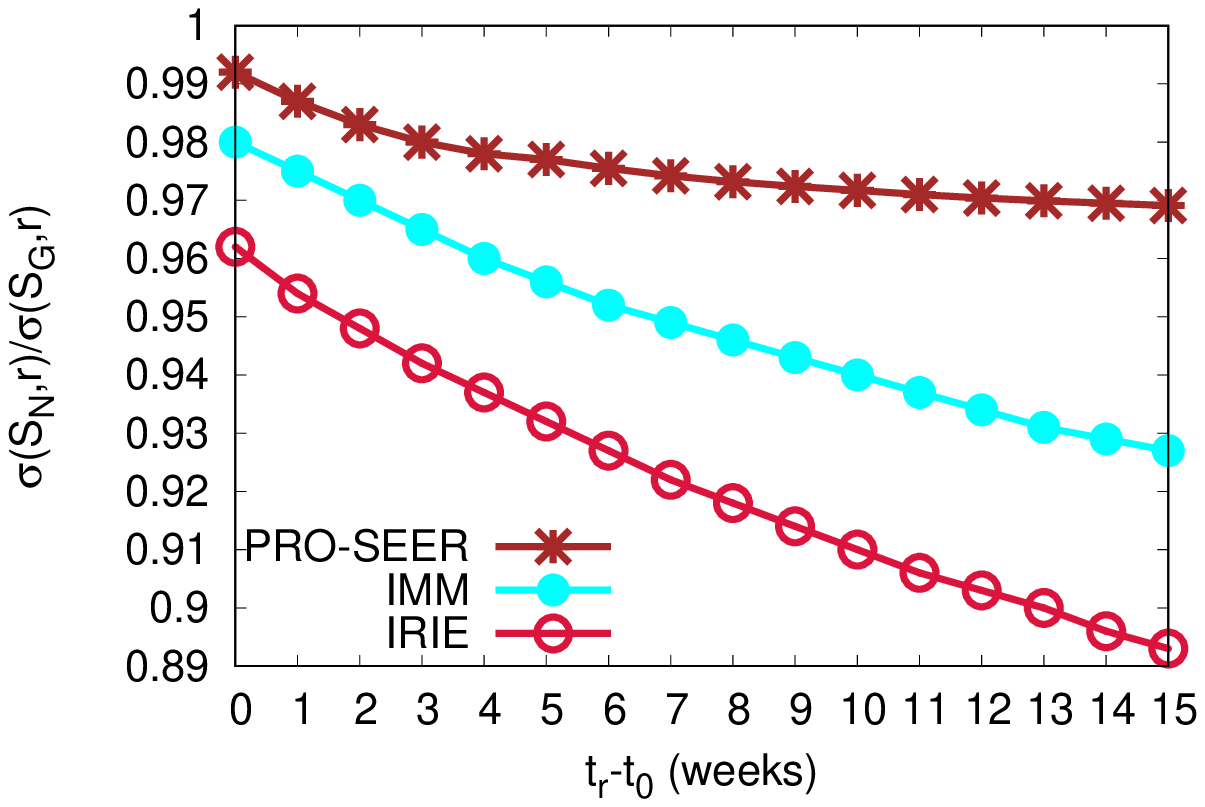, width=0.23\linewidth}}
\subfloat[][$G_0=$\textsf{Pa-$G_3$} ($k=50$)\label{f9f}]{\epsfig{file=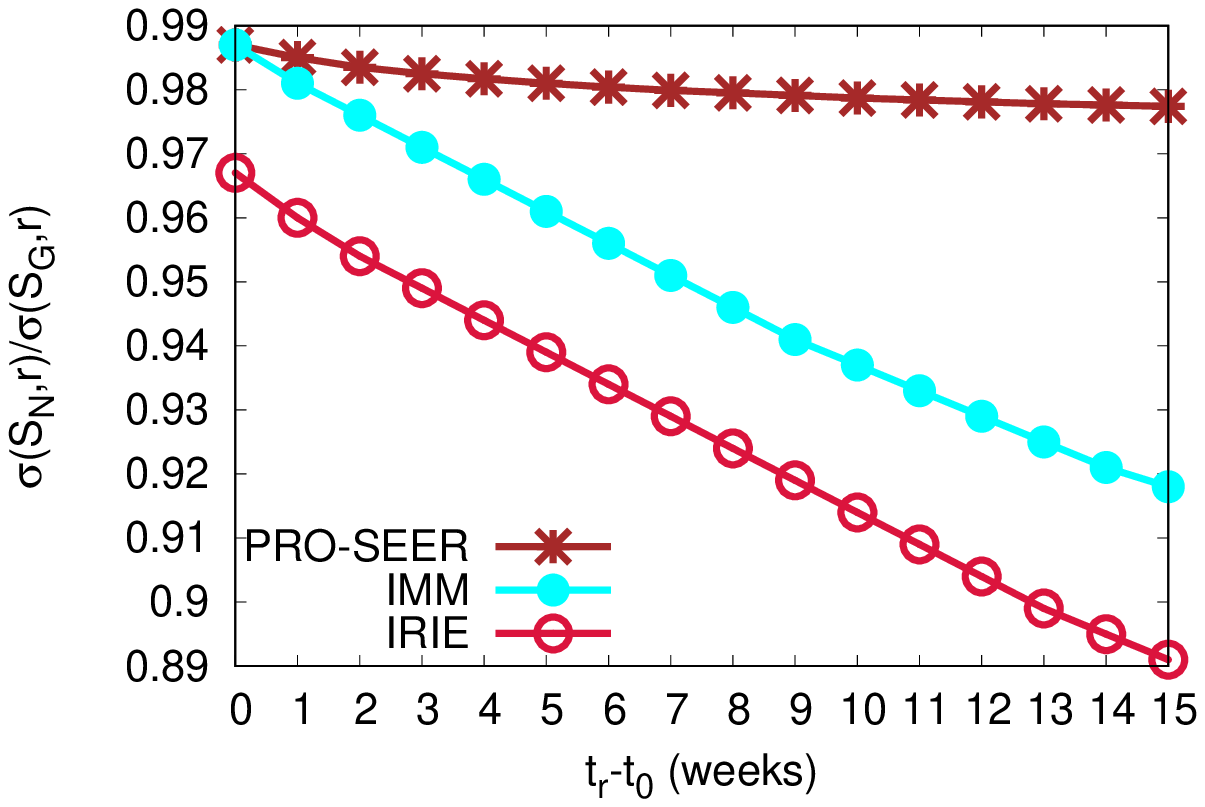, width=0.23\linewidth}}
\subfloat[][$G_0=$\textsf{Syn-$G_1$} ($k=50$)\label{f9c}]{\epsfig{file=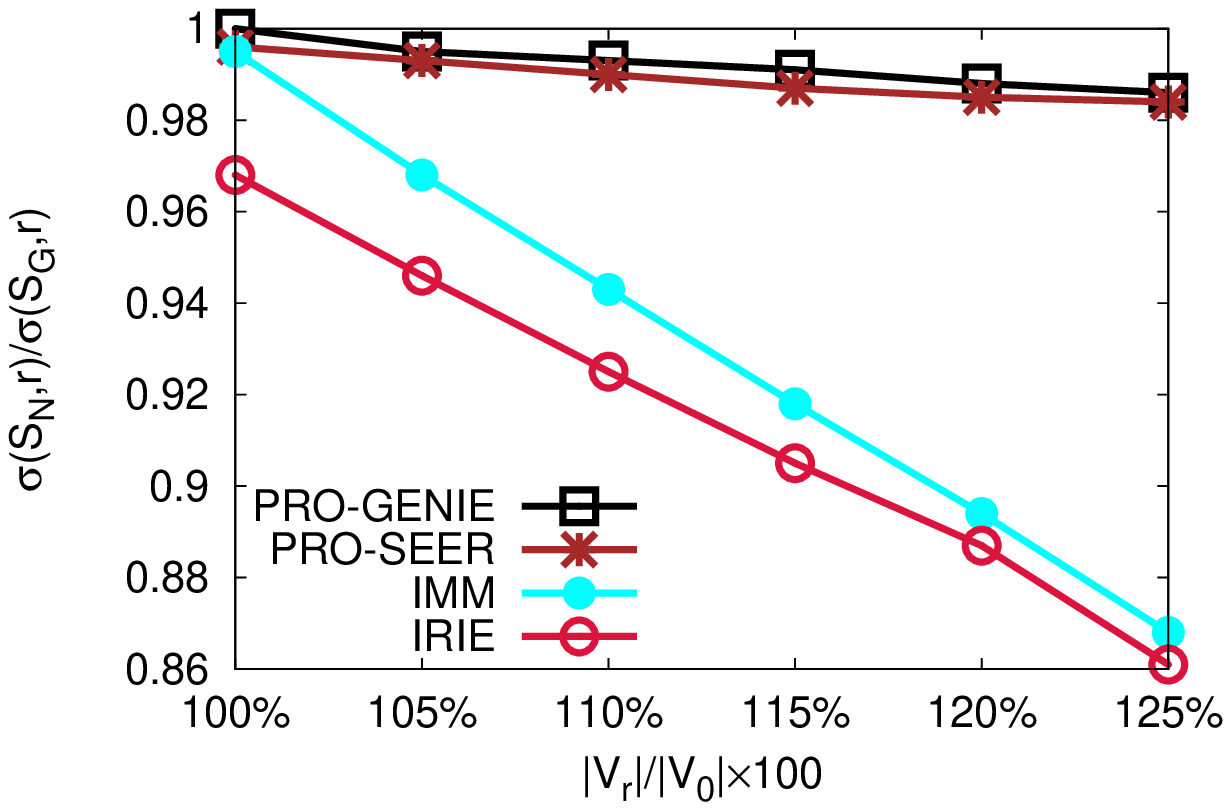, width=0.23\linewidth}}
\subfloat[][$G_0=$\textsf{Pa-$G_1$} ($k=50$)\label{f9d}]{\epsfig{file=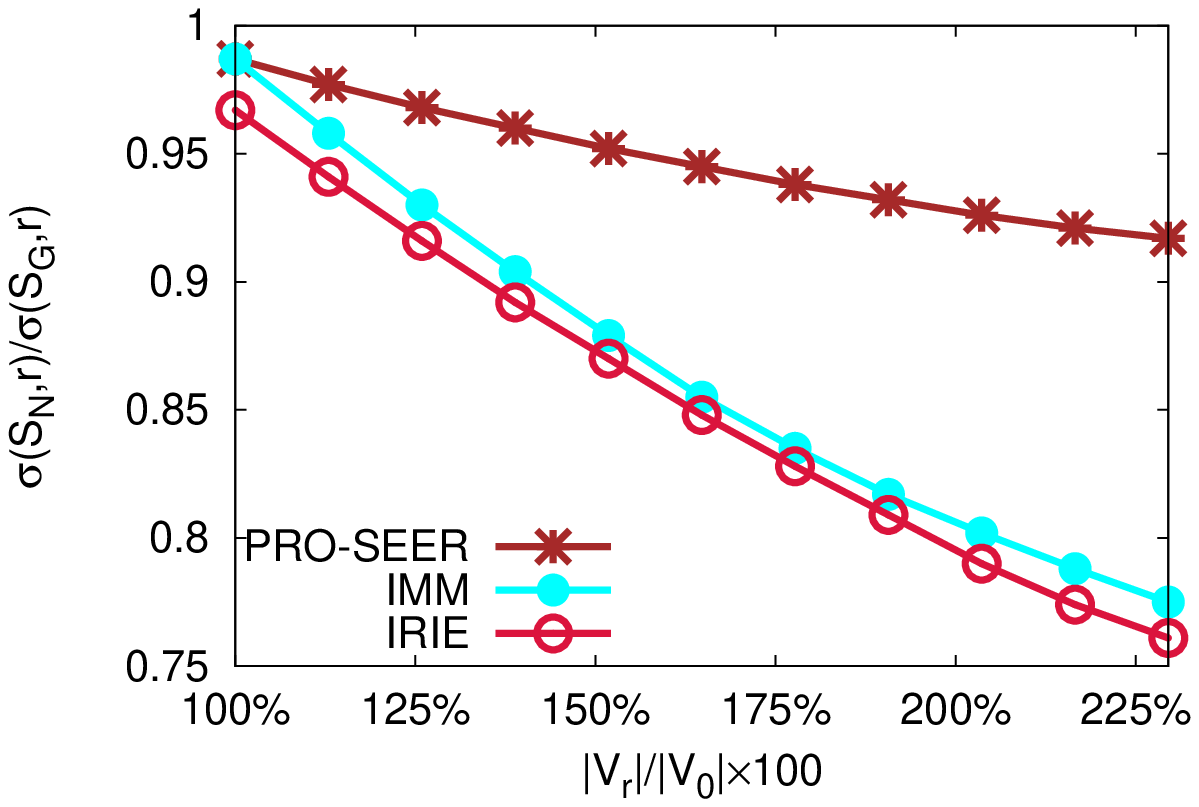, width=0.23\linewidth}}

\vspace{0ex}\caption{\small Effect of degree of change. $S_G$ and $S_N$ are the seeds selected by \emph{Greedy} on $G_r$ and by different algorithms on $G_0$, respectively.}
\label{fig9}
\hrule
\vspace{0ex}\end{figure*}

\textbf{\underline{Seeds at current and target times.}} In Section~\ref{sec:intro}, we remarked that the seeds selected at current time $t_0$ can be significantly different from the seeds selected at target time $t_r > t_0$ due to the evolution of the underlying network. Hence, we first investigate whether this is indeed true. That is, whether the seeds selected by a state-of-the-art \textsc{im} algorithm at $t_0$ differ significantly from those selected using the same algorithm at time $t_r$. To this end, we take a pair of current and target networks ($G_0$, $G_r$) at time points $t_0$ and $t_r$. We plot the ranked seed nodes in $G_r$ on the X-axis by running a classical \textsc{im} algorithm on it, which can be considered as the ground-truth seed set. For clarity, we only consider the top-10 most influential seed nodes (ranked by their expected influence) in $G_r$ that also exist in $G_0$. Specifically, in our experiments these seeds are selected by running \textit{Greedy} on $G_r$. Then, we plot on Y-axis the corresponding ranks of these seeds in $G_0$ by running the \textit{Greedy} and \textsc{proteus-genie} algorithms on $G_0$. Hence, in our plot if a seed $v$ occupies the coordinate $(a,b)$ then it means that $v$ is ranked $a$-th at time $t_0$ (\ie it exhibits the $a$-th maximal marginal expected influence in $G_0$) and ranked $b$-th in $G_r$ at time $t_r$. Consequently, the larger the deviation from $y=x$, the worse is the quality of selected seeds as the seeds selected at $t_0$ differ significantly from the seeds needed to maximize influence at $t_r$ (recall that our goal is to identify $k$ seeds at $t_0$ that maximizes influence spread at $t_r$).

Note that for networks with high degree of change, it is intuitive to expect the seeds to be different in $G_0$ and $G_r$. Hence, we use the synthetic datasets for this experiment as it exhibits low degree of change (can be considered as ``worst'' case scenario). Fig.~\ref{fig3} plots the ranks of the top-10 seeds at times $t_0$ and $t_r$ for different pairs of networks. For instance, in Fig.~\ref{fig3}(a), $G_0$ and $G_r$ are network snapshots \textsf{Syn-$G_1$} and \textsf{Syn-$G_3$}, respectively. We have the following observations. First, the ranks of seeds selected by \textit{Greedy} using $G_0$ deviates significantly from those on $G_r$ for all datasets. That is, the seeds selected by classical \textsc{im} algorithms at times $t_0$ and $t_r$ differ significantly. Consequently, seeds selected at $t_0$ may not be suitable for maximizing the influence at $t_r$ (further validated below). Second, the ranks of the seeds selected by \textsc{proteus-genie} at time $t_0$ are relatively closer to the top-10 seeds in $G_r$ for all datasets, emphasizing the need for reformulating the classical \textsc{im} problem as \textsc{proteus-im} problem.

%\begin{figure}[!th]
%\centering
%%\subfloat[][From Syn-$G_1$ to Syn-$G_3$\label{f6a}]{\epsfig{file=fig/EffectIinG1G3.eps, width=0.24\linewidth}}
%%\subfloat[][From Syn-$G_2$ to Syn-$G_3$\label{f6b}]{\epsfig{file=fig/EffectIinG2G3.eps, width=0.24\linewidth}}
%%\subfloat[][From Syn-$G_3$ to Syn-$G_4$\label{f6c}]{\epsfig{file=fig/EffectIinG3G4.eps, width=0.24\linewidth}}
%\epsfig{file=fig/EffectIinpG1pG2.eps, width=0.8\columnwidth}
%\vspace{0ex}\caption{Effect of $I$ in \textsc{new-genie} ($k=50$, from $G_0$ to $G_r$).}\label{fig6}
%%\hrule
%\vspace{0ex}\end{figure}

\textbf{\underline{Effect of degree of change and $t_r$.}} The above experiments demonstrate that the seeds at current and target times are different. We now study how the degree of change to the network impact the influence and seed set. Observe that degree of change is correlated with $t_r$. Intuitively, the longer is the influence propagation time ($t_r$) the greater is the degree of change to the network. First, we investigate the influence spread quality by varying the duration between $G_0$ and $G_r$ (\ie influence propagation time $t_r$). Since influence propagation may take weeks, we use the real-world networks to report the effect of target time $t_r$ by selecting several states of $G_r$ at different target times such that $t_r-t_0$ ranges from 0 to 15 weeks. The results are shown in Figures~\ref{fig9}(a) and (b), where $G_0$ are \textit{Pa-$G_1$} and \textit{Pa-$G_3$}, respectively. Note that the procedure to extract a $G_r$ at a specific week is same as the one to extract different snapshots in the \textit{Patents} network (Section~\ref{sec:setup}). For instance, in Fig.~\ref{fig9}(a) we fix $G_0$ as \textsf{Pa-$G_1$} and acquire each temporal state of the network at $1, 2, \ldots, 15$ weeks after \textsf{Pa-$G_1$} by selecting all citations (edges) that appear before the corresponding week. We evaluate the influence spread of the seeds selected from $G_0$ by different algorithms to those selected from different states of $G_r$ using \textit{Greedy} (run directly on $G_r$). Here $S_N$ denotes the seeds selected by different algorithms running on $G_0$ while $S_G$ denotes the ideal solution by running \emph{Greedy} in $G_r$. Therefore, the Y-axis shows the ratio that compares the expected influence for different algorithms to that of running \emph{Greedy} in $G_r$. When $G_r=G_0$ (\ie $t_r=t_0$), the problem degenerates to classical \textsc{im}. Consequently, seeds of all algorithms share almost the same quality. However, as $t_r$ increases, the quality of seeds set selected by different algorithms at $G_0$ decreases. Clearly, in comparison to classical \textsc{im} techniques, our \textsc{proteus-seer} exhibits the highest influence spread quality for $t_r > 0$.

\begin{figure*}[t]
\centering
%\subfloat[][\label{f4a}]{\epsfig{file=fiinfsprd10-50.eps, width=0.2\linewidth}}
\subfloat[][Low DoC\label{f4b}]{\epsfig{file=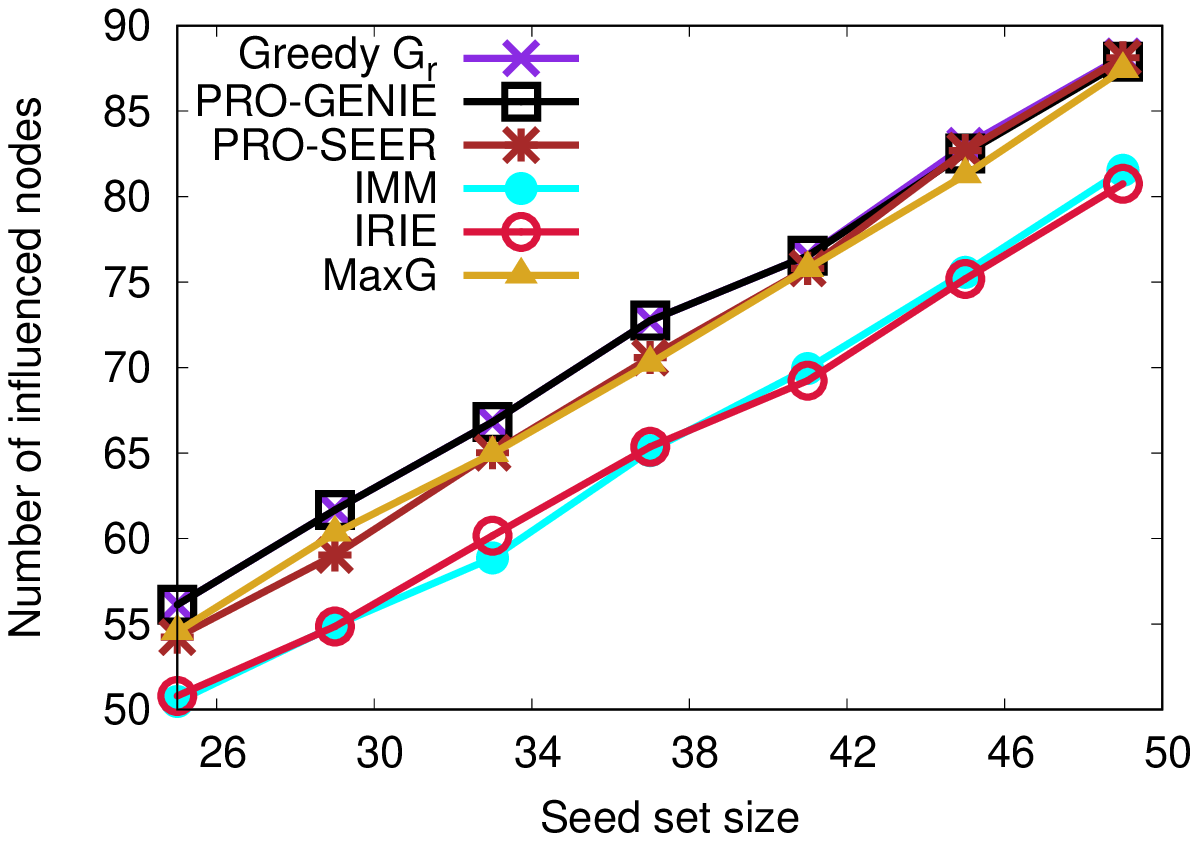, width=0.23\linewidth}}
%\subfloat[][$G_0=$Syn-$G_3$, $G_r=$Syn-$G_4$\label{f4c}]{\epsfig{file=figinfsprd50-100.eps, width=0.22\linewidth}}
%\subfloat[][$G_0$=\textit{Pa-$G_1$}, $G_r$=\textit{Pa-$G_3$}\label{f4a}]{\epsfig{file=fig/infsprdpa1pa3.eps, width=0.16\linewidth}}
\subfloat[][High DoC\label{f4d}]{\epsfig{file=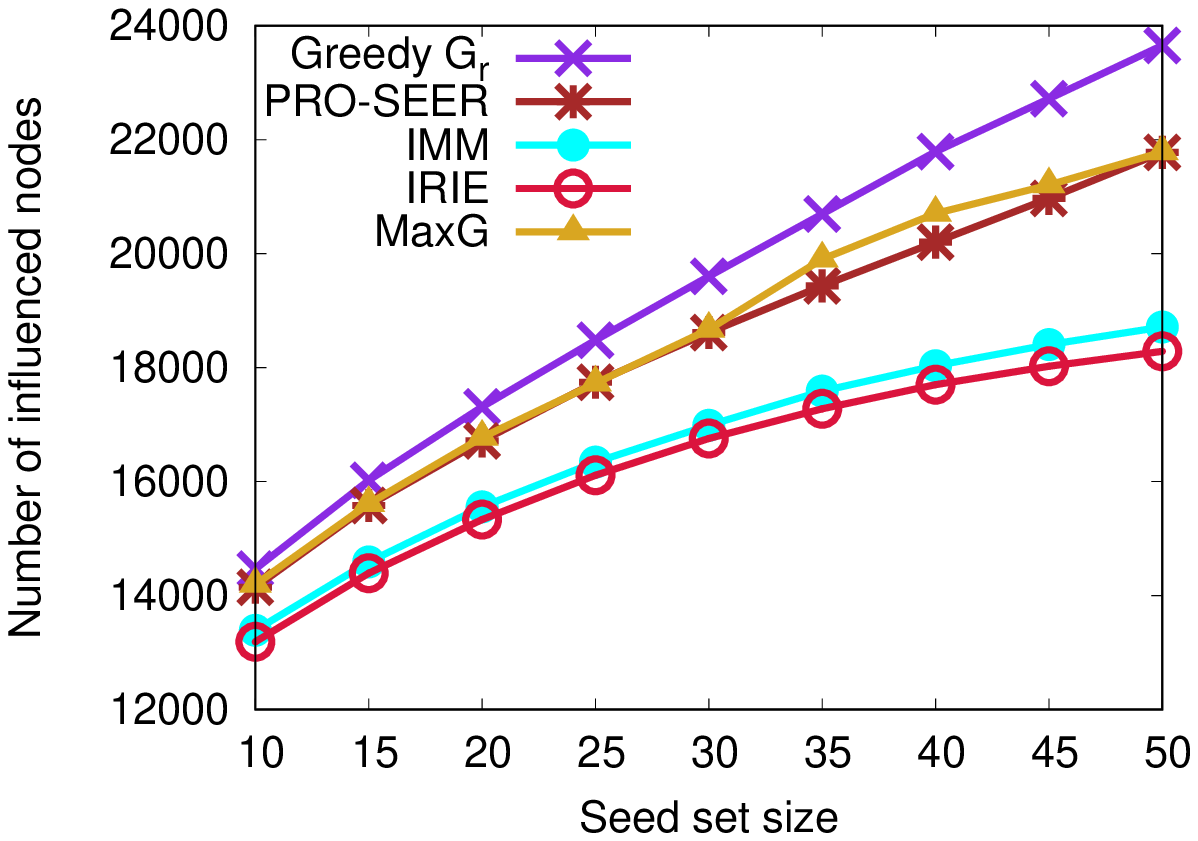, width=0.23\linewidth}}
\subfloat[][High DoC\label{f4c}]{\epsfig{file=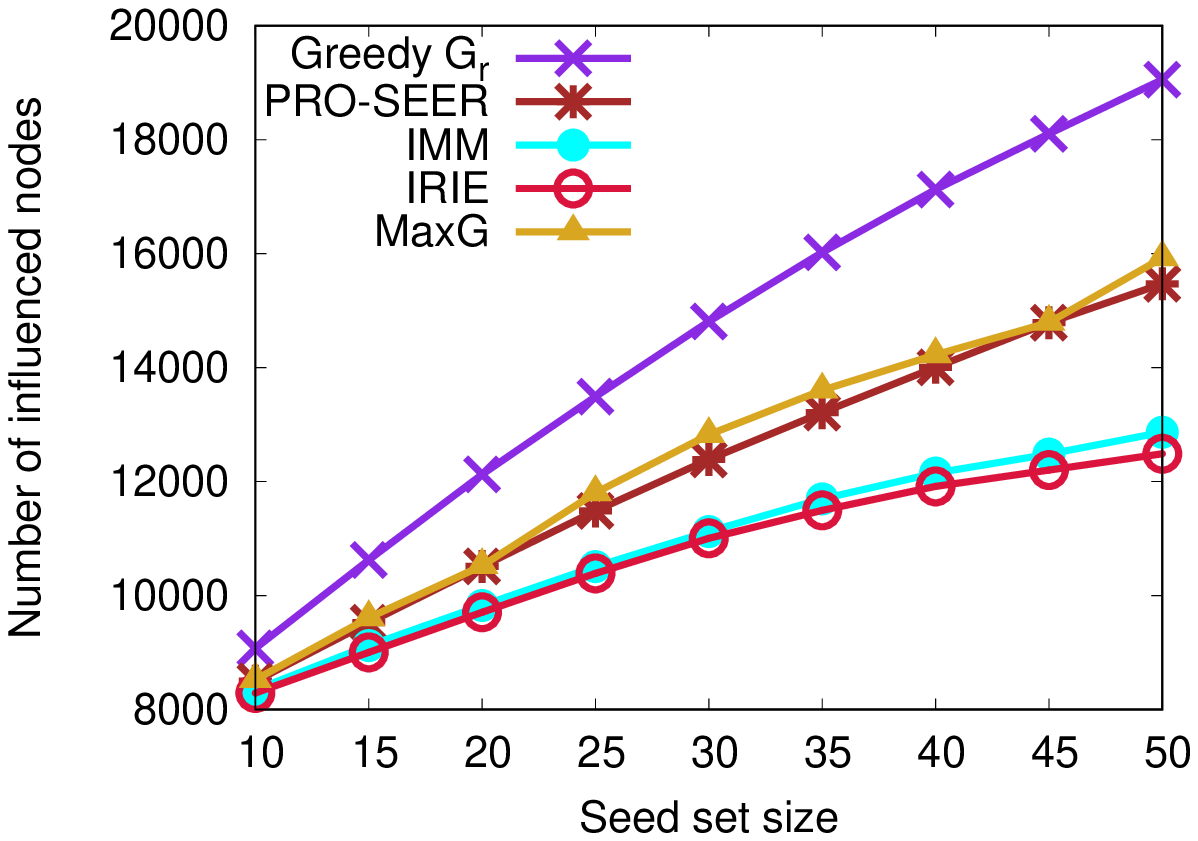, width=0.23\linewidth}}
\subfloat[][Moderate DoC\label{f4e}]{\epsfig{file=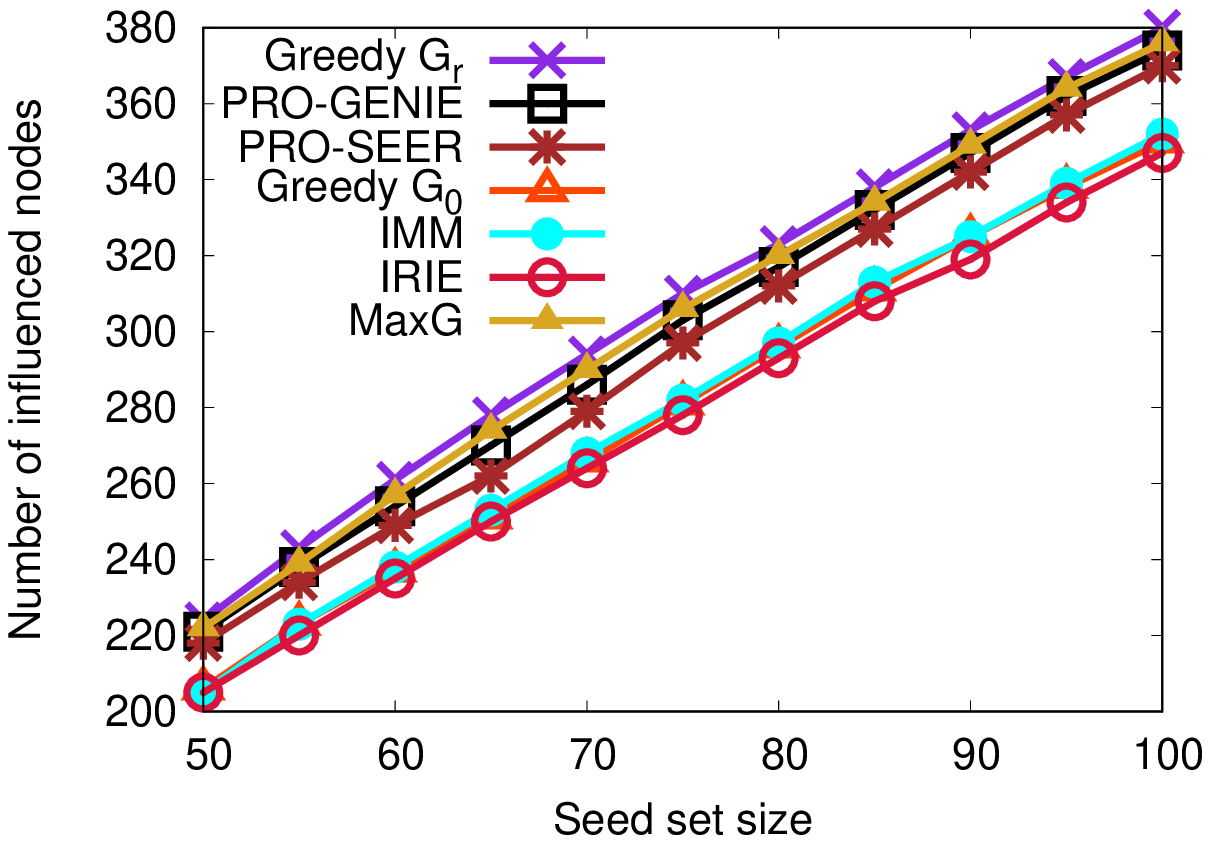, width=0.23\linewidth}}
\vspace{0ex}\caption{Influence spread of different seeds sets. (a) $G_0$ = \textsf{Syn-$G_2$}, $G_r$ = \textsf{Syn-$G_3$}; (b) $G_0$ = \textsf{Pa-$G_3$}, $G_r$ = \textsf{Pa-$G_4$}; (c) $G_0$ = \textsf{Pa-$G_1$}, $G_r$ = \textsf{Pa-$G_3$}; (d) $G_0$ = \textsf{Ph-$G_1$}, $G_r$ = \textsf{Ph-$G_2$}.}\label{fig4}
\hrule
\vspace{0ex} \end{figure*}

Next, we vary the effect network change rate by selecting several states of $G_r$ such that $1 \leq |V_r|/|V_0| \leq 2.25$. Note that this simulates different degree of change to the topology of the network $G_0$ at $t_r$. We evaluate the influence spread of the seeds selected from $G_0$ by different algorithms to those selected from different state of $G_r$ using \textit{Greedy} (run directly over $G_r$). The results are shown in Figures~\ref{fig9}(c) and (d). When $G_r=G_0$ (\ie $|V_r|/|V_0|=1$), the problem degenerates to classical \textsc{im}. As the difference between $G_r$ and $G_0$ increases, the quality of seeds set selected by different algorithms at $G_0$ decrease. Clearly, our proposed techniques exhibit the highest influence spread quality.

\eat{We fix $G_r$ as \textsf{Syn-$G_3$}, and run \textsc{proteus-seer} on \textsf{Syn-$G_1$} and \textsf{Syn-$G_2$} (\ie $G_0$). \eat{The influence spread quality is reported in Fig.~\ref{fig9}(a). Obviously, running our algorithm at a nearer temporal state, namely \textit{Syn-$G_2$}, will generate better quality seeds. Similar phenomenon is also observed in real-world dataset \emph{Pa-$G_1$} to \emph{Pa-$G_4$} (Fig.~\ref{fig9}(b)).} In our experiments over synthetic dataset, our algorithm still performs at least 13\% better comparing to the best competitors when there is at least 5\% difference in the edge set between $G_0$ and $G_r$.}

\textbf{\underline{Influence spread for different $k$.}} Next, we simulate the influence spread of selected seeds for networks with varying $k$ and investigate whether state-of-the-art \textsc{im} algorithms exhibit similar or different influence spread quality compared to our proposed algorithms for the \textsc{proteus-im} problem. Fig.~\ref{fig4} plots the influence spreads (with influence probability $p=0.01$) for different $k$ for networks exhibiting different DoC. In each of the figures, we select top-$k$ seeds in $G_0$ using \emph{Greedy}, \textsc{irie}, \textsc{imm}, \emph{MaxG}, \textsc{proteus-genie}, and \textsc{proteus-seer} and simulate the influence spread process in $G_r$. The influence spread is measured by the number of eventually influence nodes that is averaged over 10,000 simulations. We compare the influence spread results with that of the seeds selected using \emph{Greedy} in $G_r$, which can be viewed as the ground-truth seeds set. Note that closer the influence spread (computed by a specific technique) is to this ground-truth seed set, the better is its influence spread quality. \eat{Similarly, the work in~\cite{DBLP:conf/sdm/AggarwalLY12} is also impractical, as it works (\ie select $k$ seeds in $G_0$ at $t_0$) only when $G_0, \ldots, G_r$ are known, which is impossible. What's more, it cannot provide ideal solution as ``\emph{Greedy} within $G_r$''. Therefore, we select to adopt ``\emph{Greedy} within $G_r$'' in the experimental study. On the other hand, ~\cite{DBLP:conf/icdm/ZhuangSTZS13} is based on strict assumptions as discussed in Section~\ref{sec:relwork}, which our scenario or dataset cannot satisfy, it cannot be applied in our problem settings.} Specifically, $(G_0, G_r)$ are chosen to represent different degree of change (DoC).

Observe that \textsc{proteus-genie} and \emph{MaxG} achieve the best influence spread quality, followed by our heuristic approach \textsc{proteus-seer}. Notably, \emph{MaxG} iteratively updates the selected seeds whenever a new node arrives in the network.  Interestingly, despite the impractical assumptions made by \textit{MaxG} as mentioned in Section~\ref{sec:setup}, it cannot provide distinguishable performance benefit compared to our algorithms. \textit{In other words, \textsc{proteus-genie} and \textsc{proteus-seer} demonstrate comparable seed set quality without  assuming the knowledge of complete topology of the target network (unlike MaxG). }
In summary, the influence spread of our proposed approaches are within 83\% - 99\% of the ideal solution. In contrast, the state-of-the-art classical \textsc{im} approaches only achieve 65\%-89\% of the ideal influence spread.

\begin{figure*}[t]
\centering
\vspace{0ex}
\subfloat[][$G_0=$ \textsf{Syn-$G_1$}, $G_r=$ \textsf{Syn-$G_3$}\label{f5a}]{\epsfig{file=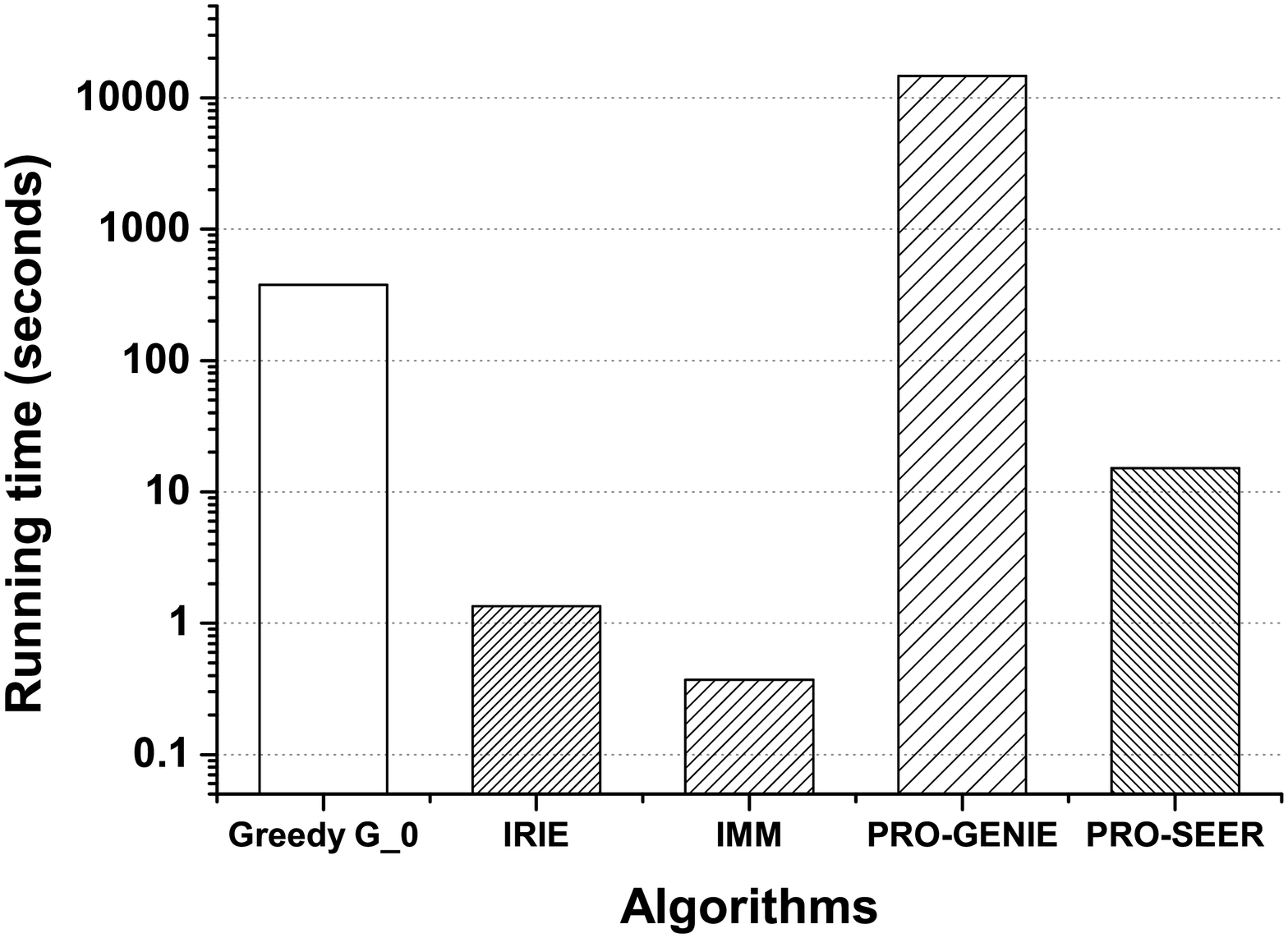, width=0.23\linewidth}}
\subfloat[][$G_0=$ \textsf{Ph-$G_1$}, $G_r=$ \textsf{Ph-$G_2$}\label{f5b}]{\epsfig{file=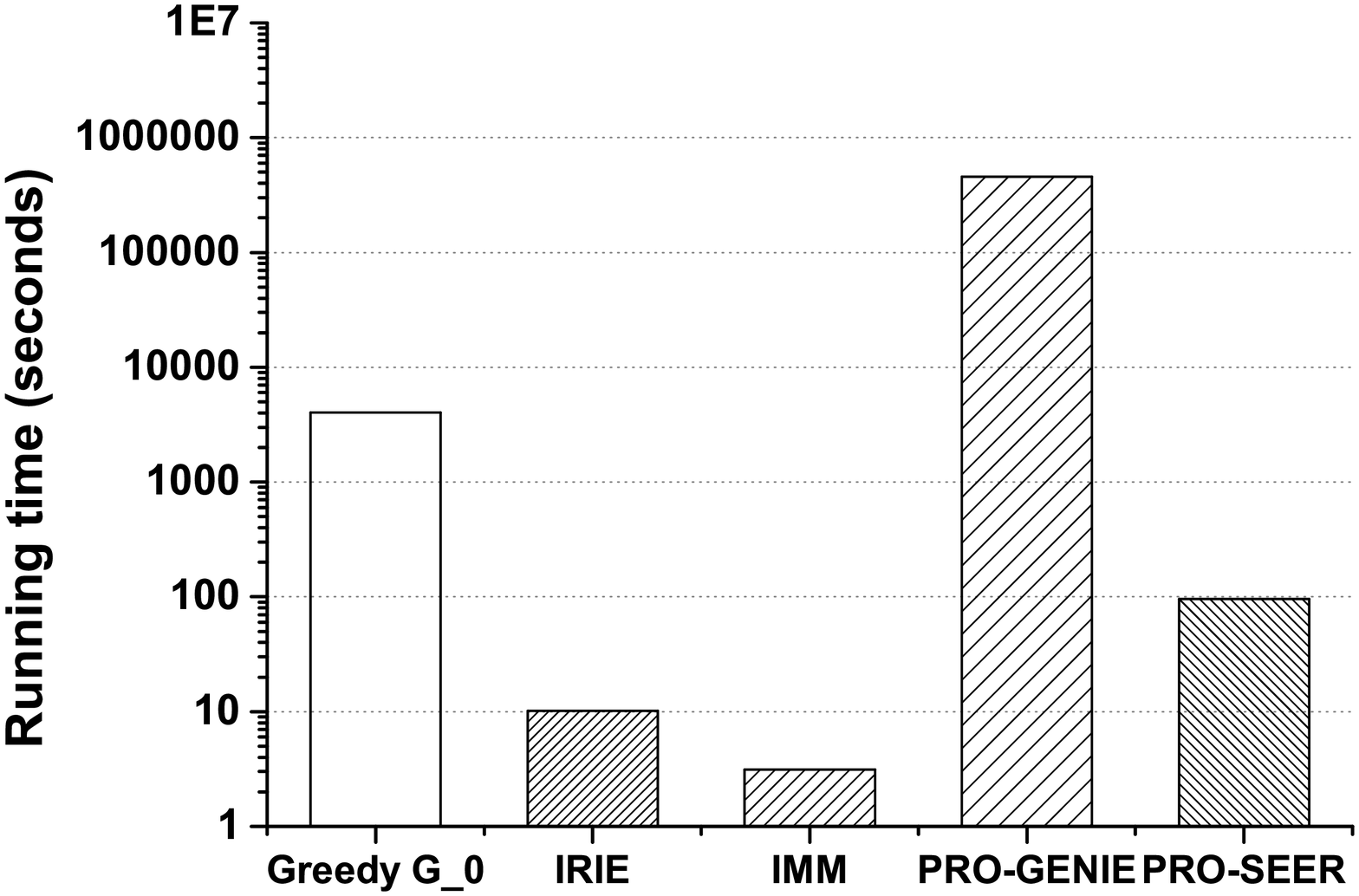, width=0.23\linewidth}}
\subfloat[][$G_0=$\textsf{Pa-$G_1$}, $G_r=$\textsf{Pa-$G_3$}\label{f8b}]{\epsfig{file=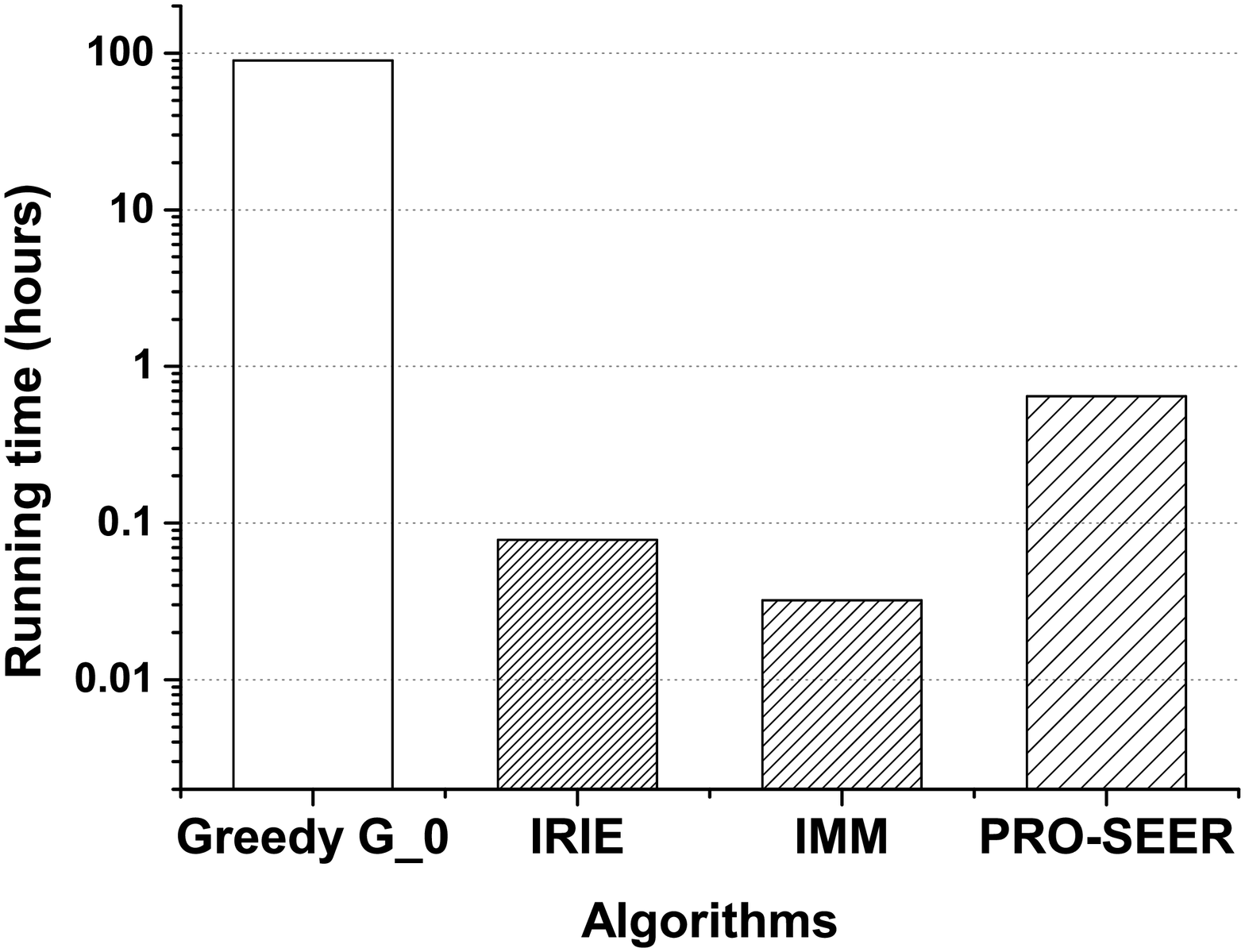, width=0.23\linewidth}}
\subfloat[][$G_0=$\textsf{Pa-$G_3$}, $G_r=$\textsf{Pa-$G_4$}\label{f8d}]{\epsfig{file=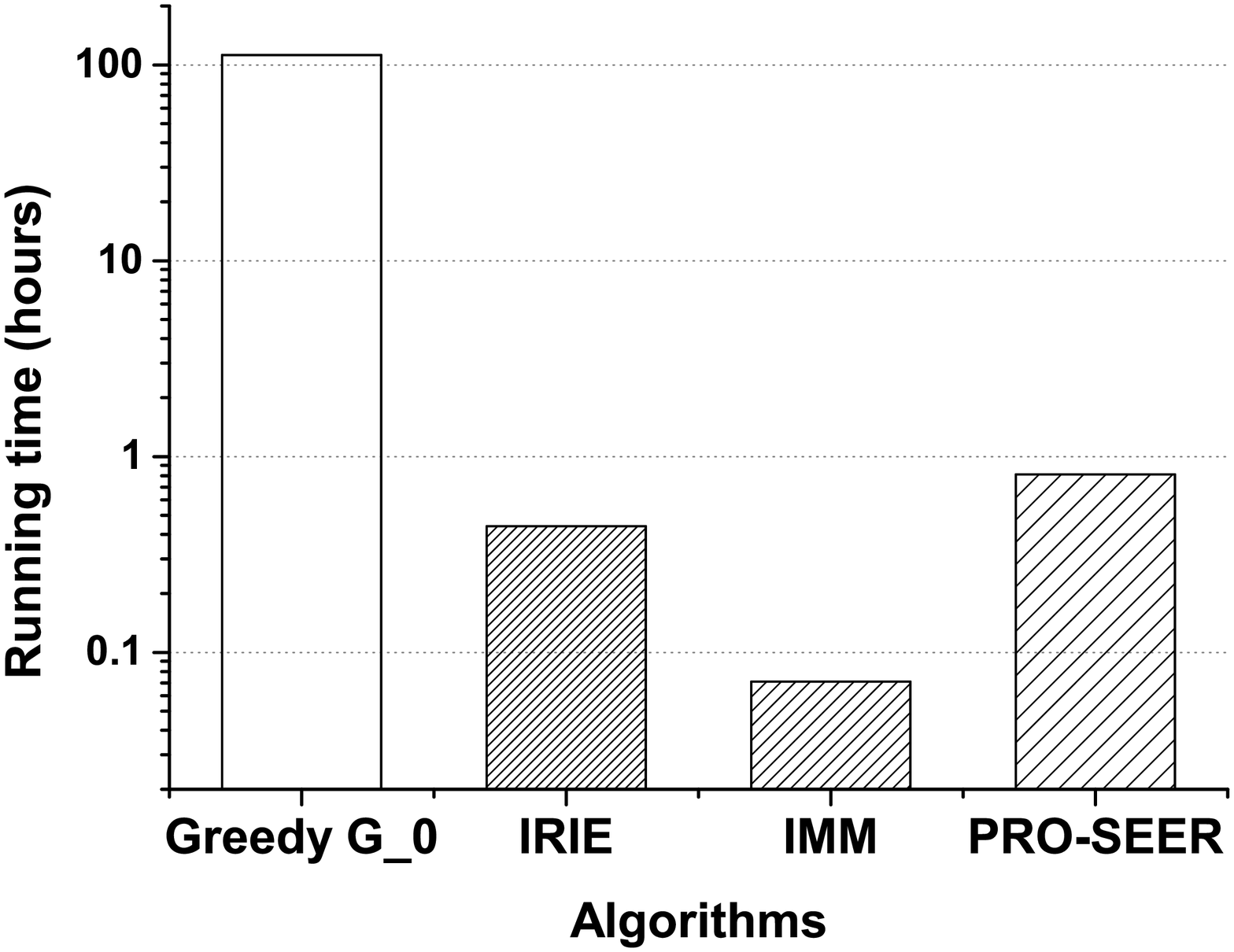, width=0.23\linewidth}}
\vspace{0ex}\caption{Running times of different algorithms.}\label{fig5}

%\hrule
%\subfloat[][From Syn-$G_1$ to Syn-$G_3$\label{f5a}]{\epsfig{file=fig/runtime10-50.eps, width=0.24\linewidth}}
%\subfloat[][From Syn-$G_2$ to Syn-$G_3$\label{f5b}]{\epsfig{file=fig/runtime20-50.eps, width=0.24\linewidth}}
%\subfloat[][From Syn-$G_3$ to Syn-$G_4$\label{f5c}]{\epsfig{file=fig/runtime50-100.eps, width=0.24\linewidth}}
%\subfloat[][From Ph-$G_1$ to Ph-$G_2$\label{f5d}]{\epsfig{file=fig/runtime32-38.eps, width=0.24\linewidth}}
\hrule
\vspace{0ex}\end{figure*}

%A keen reader may observe that compared to the phenomenon reported in Figures~\ref{fig4}(a),(b), and (d), the gap between the ideal solution and state-of-the-art \textsc{im} solutions as well as ours is enlarged in Fig.~\ref{fig4}(c). This is because the topological difference between \textit{Syn-$G_3$} and \textit{Syn-$G_4$} is significantly large compared to other cases. Particularly, there are 5,000 new nodes in \textit{Syn-$G_4$}, indicating that the \textit{Syn-$G_3$} grew twice in size during $t_0$ and $t_r$. This indicates that a larger topological difference between $G_0$ and $G_r$ leads to larger gap between the influence spread of (a) state-of-the-art \textsc{im} algorithms and our proposed approaches; (b) the ``ideal'' influence spread and those generated by our approaches.

\textbf{\underline{Running times.}} Fig.~\ref{fig5} reports the running times of different algorithms for different DoC. Specifically, we run \emph{Greedy}, \textsc{irie}, \textsc{imm}, \textsc{proteus-genie}, and \textsc{proteus-seer} on $G_0$\footnote{ \emph{MaxG} keeps on running during the evolution of a network in contrast to all other competitors. Hence, for fair comparison its running time is not included.} for the three different datasets. Observe that although \textsc{proteus-genie} produces most accurate results, it also consumes the longest time as it requires $I$ iterations of network evolution simulations. On the other hand, \textsc{proteus-seer} is significantly faster than \textsc{proteus-genie} as the former avoids huge number of iterations caused by $R$ and $I$.  Our \textsc{proteus-seer} finishes within an hour on the largest network while providing near-optimal influence spread quality. Therefore, \textsc{proteus-seer} is suitable for time-sensitive tasks and gives a good balance between influence spread quality and running time. Note that although our techniques are slower than \textsc{irie} and \textsc{imm}, as reported earlier, these approaches have poorer influence spread quality. It is important to reemphasize that the seed set quality is paramount to companies as they would like to maximize the influence spread of their products.

\begin{figure}[t]

\centering
\subfloat[][Low DoC \label{f5c}]{\epsfig{file=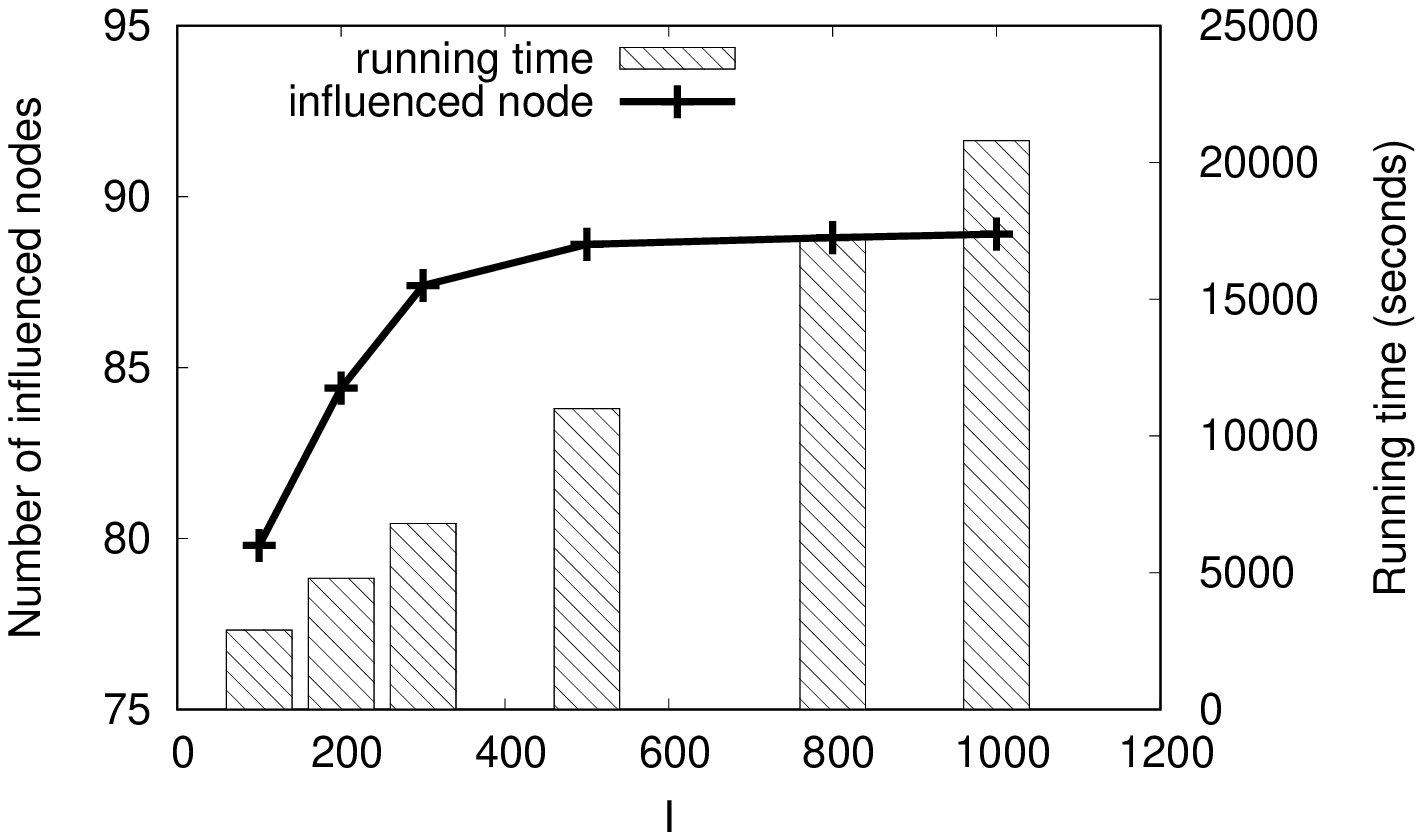, width=0.45\linewidth}}
\hspace{2ex}\subfloat[][Moderate DoC]{\epsfig{file=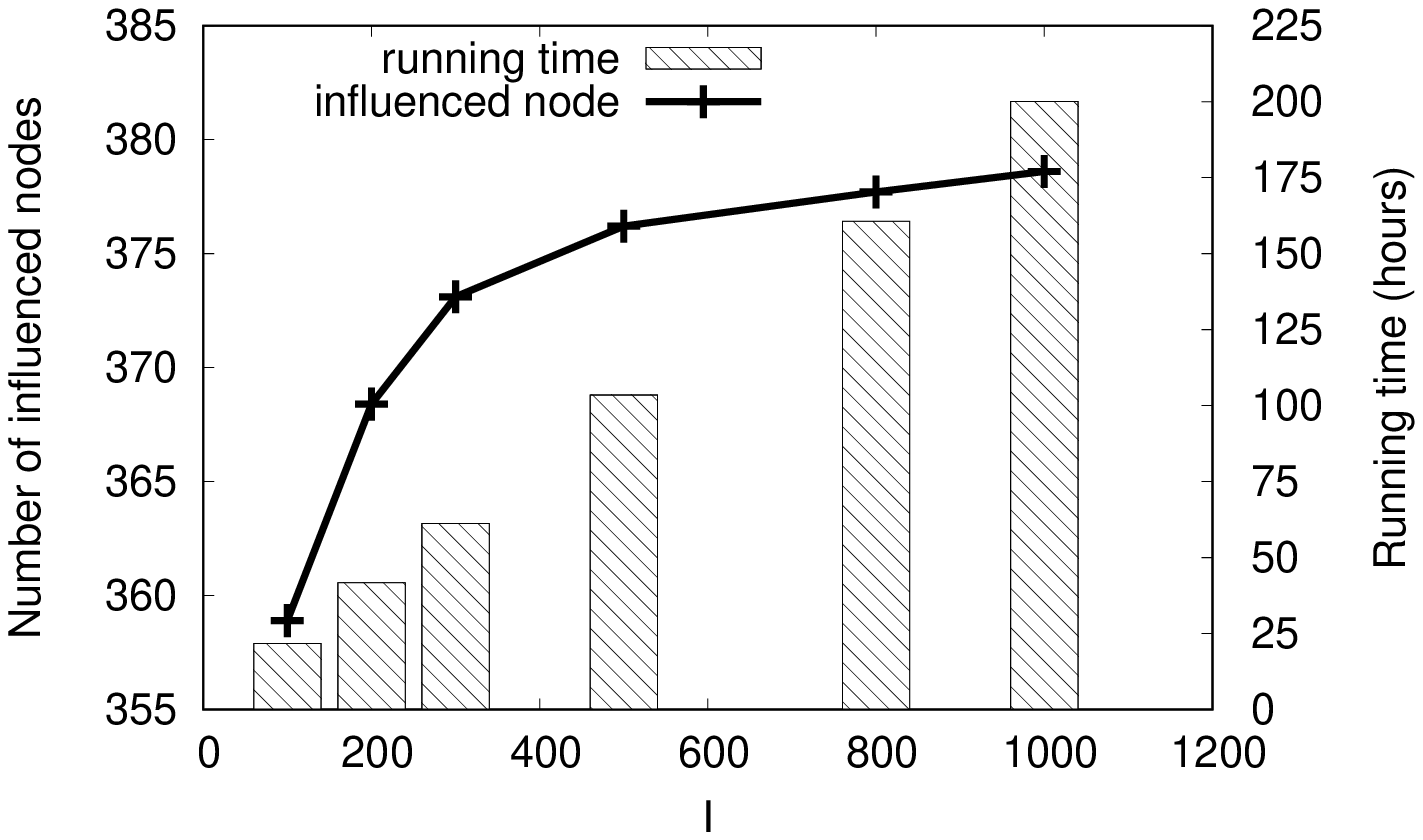, width=0.45\linewidth}}
\vspace{0ex}\caption{Effect of $I$ ($k=50$). (a) $G_0=$ \textsf{Syn-$G_1$}, $G_r=$ \textsf{Syn-$G_2$}; (b) $G_0=$ \textsf{Ph-$G_1$}, $G_r=$ \textsf{Ph-$G_2$}.}\label{Ieffect}
\hrule
\vspace{0ex}\end{figure}

\eat{\begin{figure*}[t]
\centering

\subfloat[][$G_0=$Pa-$G_1$, $G_r=$Pa-$G_3$\label{f8a}]{\epsfig{file=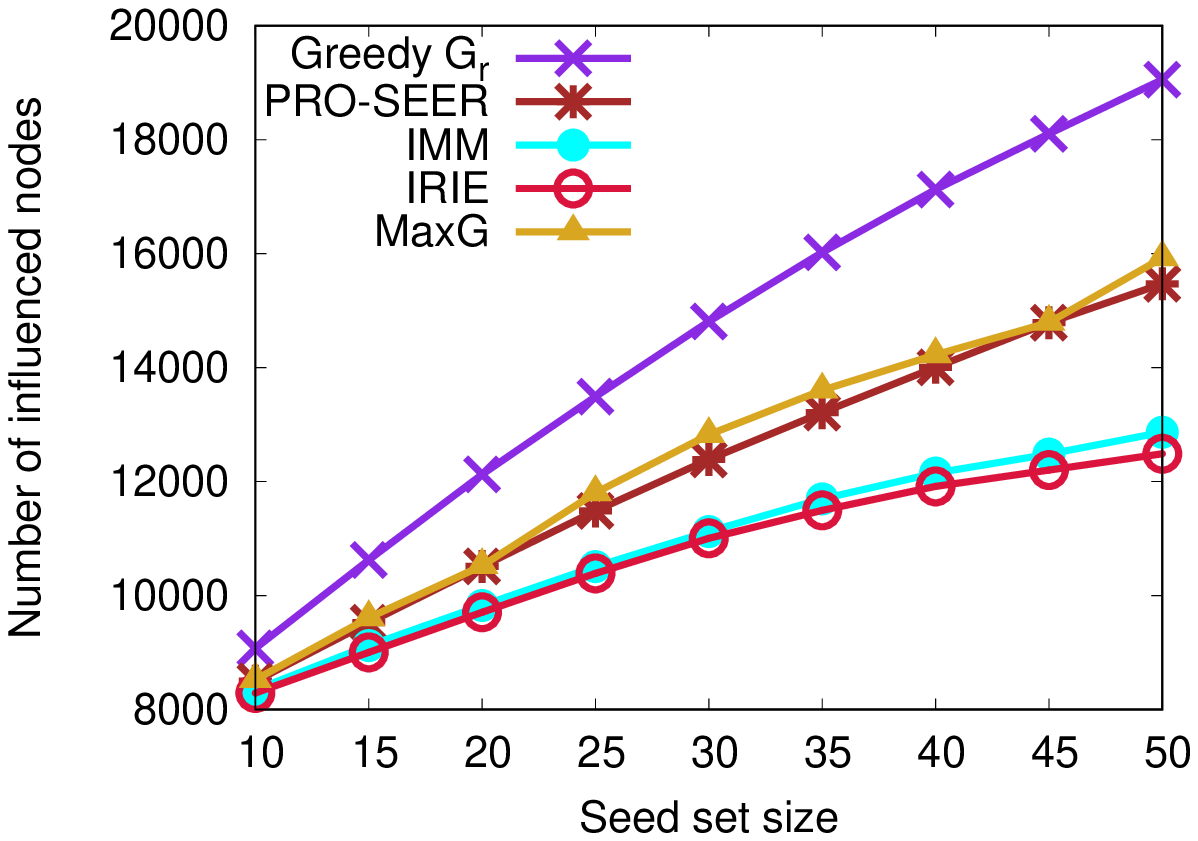, width=0.23\linewidth}}
\subfloat[][$G_0=$Pa-$G_3$, $G_r=$Pa-$G_4$\label{f8c}]{\epsfig{file=infsprdpa2pa3.eps, width=0.23\linewidth}}
\vspace{0ex}\caption{Comparison in \emph{Patents} dataset: (a,c) influence spread; (b,d) running time}\label{fig8}
\hrule
\vspace{0ex}
\end{figure*}

\textbf{\underline{Application in Large Networks.}} Next, we investigate the performance of our algorithms on a real-world large network, \emph{Patents}, over the temporal states as described in last four rows of Table~\ref{t:data}. We compare the performances of \textsc{proteus-seer}\footnote{ We do not report \textsc{proteus-genie} as it does not finish within 5 days.}, \textsc{imm}, and \textsc{irie}  (all on $G_0$) with \emph{Greedy} on $G_r$ (\ie ideal seed set). The results are shown in Fig.~\ref{fig8}. In particular, the influence spread of \textsc{proteus-seer} is within 80\% and 83\% of the ideal one in Fig.~\ref{fig8}(a) and (c), respectively. More importantly, the influence spread quality is once again consistently superior to state-of-the-art \textsc{im} techniques. Its performance is comparable to \textit{MaxG} due to the reasons mentioned above.}

\textbf{\underline{Effect of $I$.}} Intuitively, if we increase the number of instances of predicted target network (\ie $I$), it may increase the running time. Fig.~\ref{Ieffect} reports the running times by varying $I$ in \textsc{proteus-genie}. Obviously, the running time increases almost linearly with respect to $I$, which is consistent with Theorem~\ref{the:alg1comp}. We also investigated the influence spread quality by varying $I$. Clearly, as $I$ increases the quality of selected seeds also improve. However, the improvement also follows a diminishing return pattern. If $I$ exceeds $500$, the improvement in seeds quality is within 1\%. Note that this phenomenon is favorable to our framework as we do not need to set $I$ to a very large value in order to find superior quality seeds.

\begin{figure*}[t]
\centering
\subfloat[][$G_0=$\textsf{Pa-$G_1$}, $G_r=$\textsf{Pa-$G_3$} (High DoC)\label{f10a}]{\epsfig{file=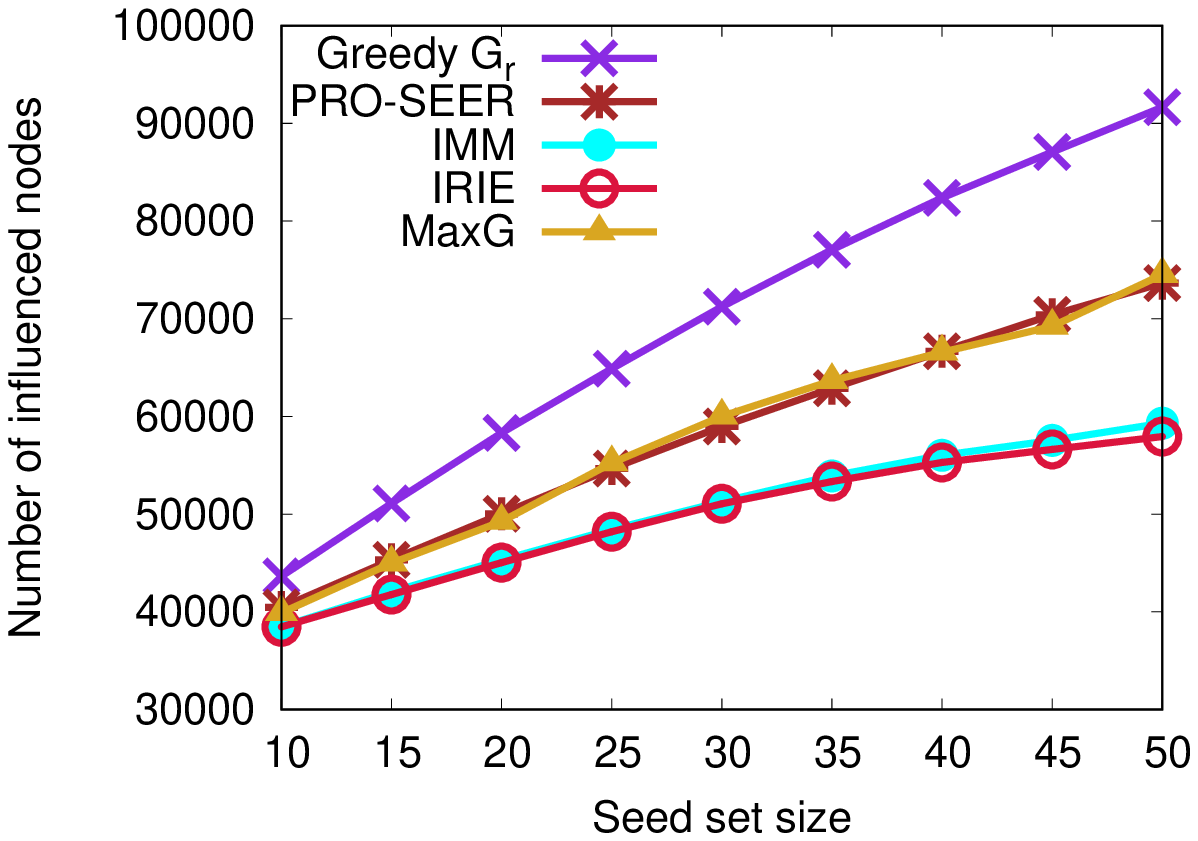, width=0.32\linewidth}}
%\subfloat[][$G_0=$\textsf{Pa-$G_2$}, $G_r=$\textsf{Pa-$G_3$}\label{f10b}]{\epsfig{file=fig/infsprdpa2pa3-05.eps, width=0.23\linewidth}}
\subfloat[][$G_0=$\textsf{Pa-$G_3$}, $G_r=$\textsf{Pa-$G_4$} (High DoC)\label{f10c}]{\epsfig{file=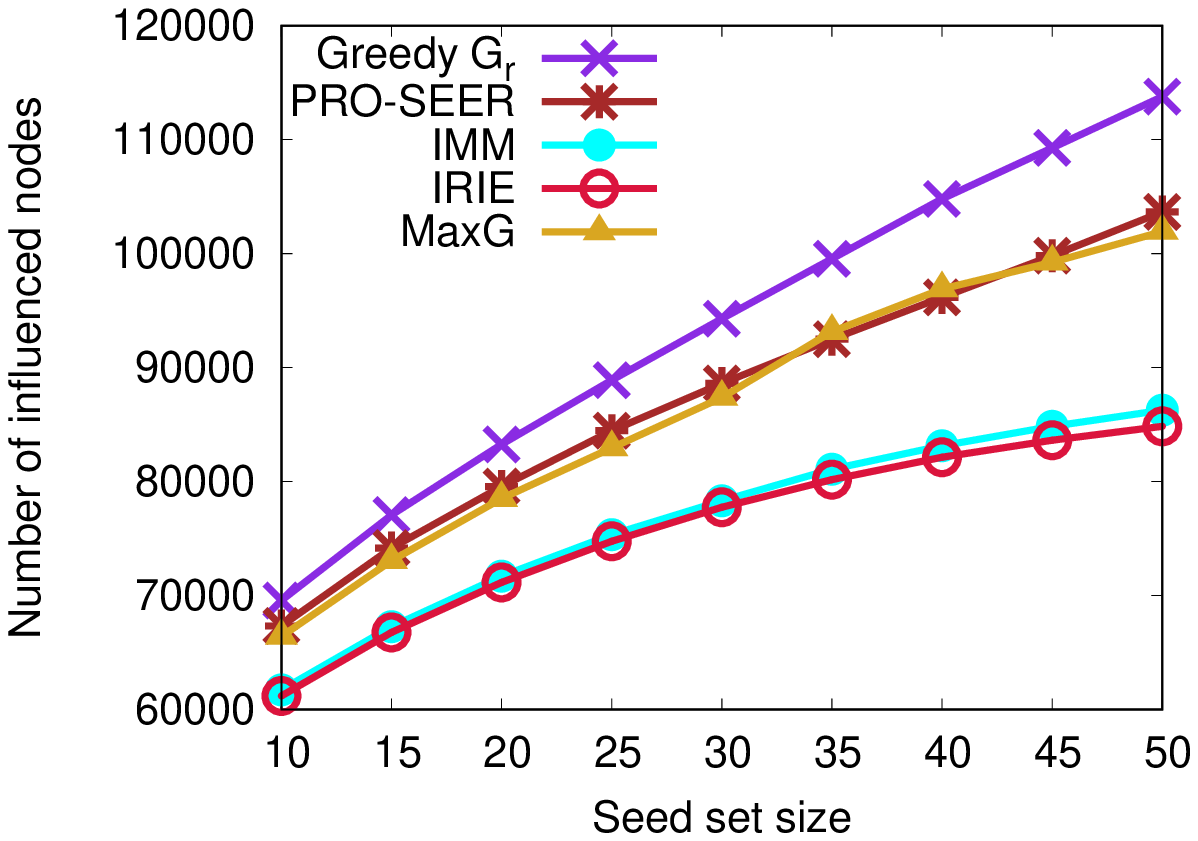, width=0.32\linewidth}}
\subfloat[][$G_0=$\textsf{Ph-$G_1$}, $G_r=$\textsf{Ph-$G_2$} (Moderate DoC)\label{f10d}]{\epsfig{file=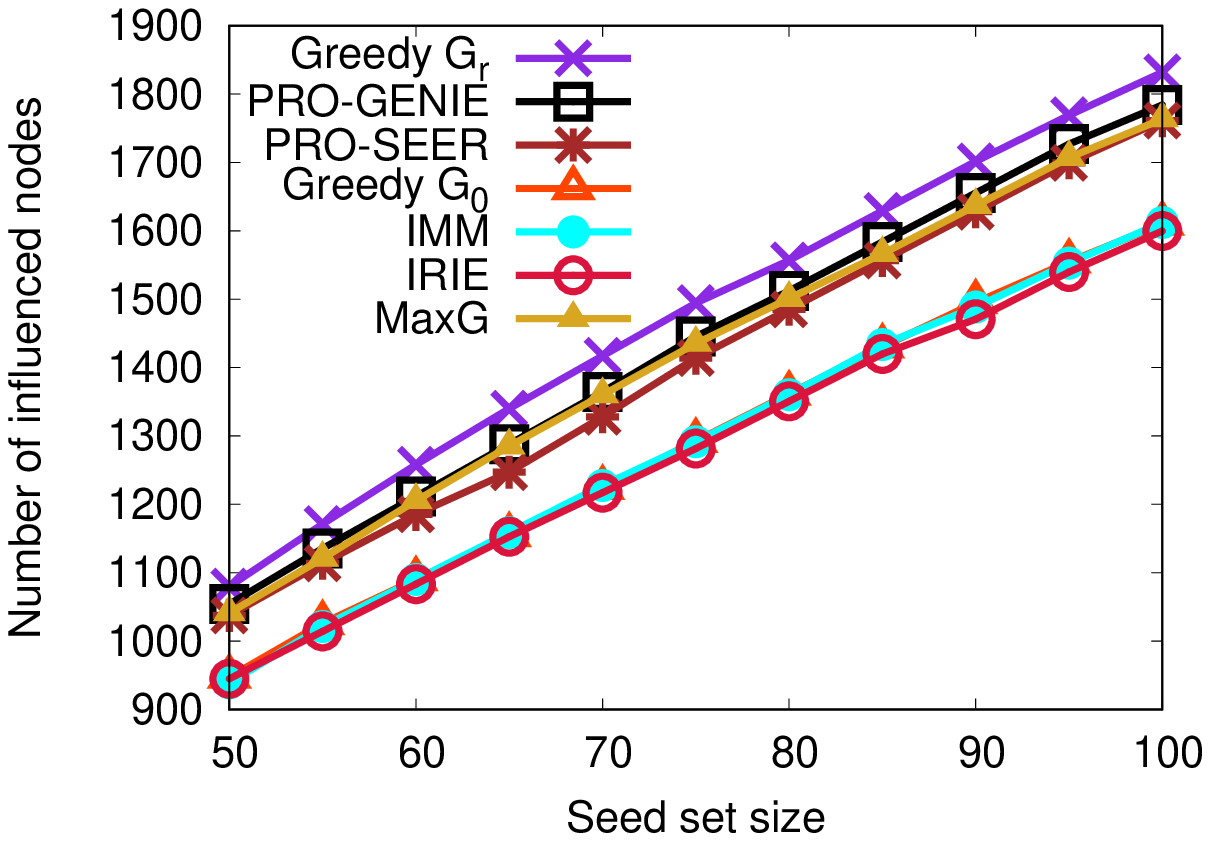, width=0.32\linewidth}} \\\vspace{0ex}

\subfloat[][$G_0=$\textsf{Pa-$G_1$}, $G_r=$\textsf{Pa-$G_3$} (High DoC)\label{f10e}]{\epsfig{file=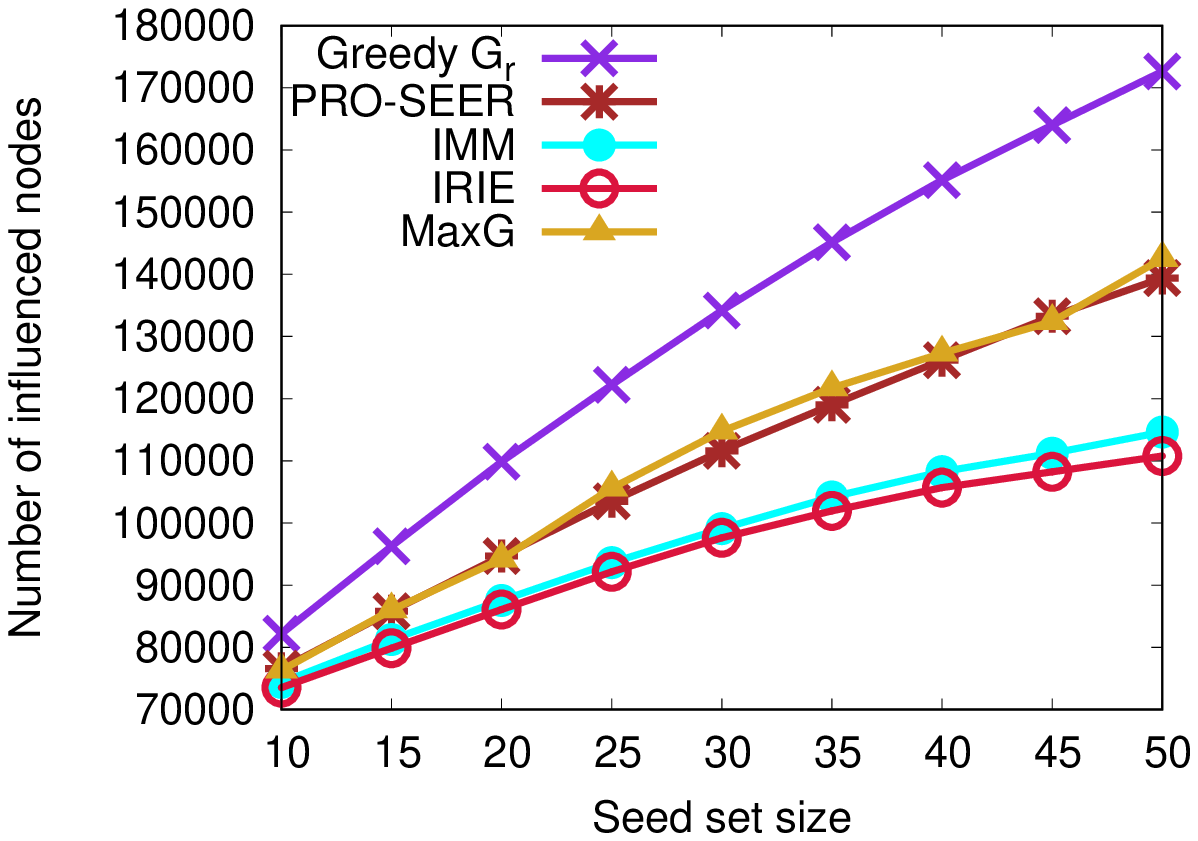, width=0.32\linewidth}}
%\subfloat[][$G_0=$\textsf{Pa-$G_2$}, $G_r=$\textsf{Pa-$G_3$}\label{f10f}]{\epsfig{file=fig/infsprdpa2pa3-10.eps, width=0.23\linewidth}}
\subfloat[][$G_0=$\textsf{Pa-$G_3$}, $G_r=$\textsf{Pa-$G_4$} (High DoC)\label{f10g}]{\epsfig{file=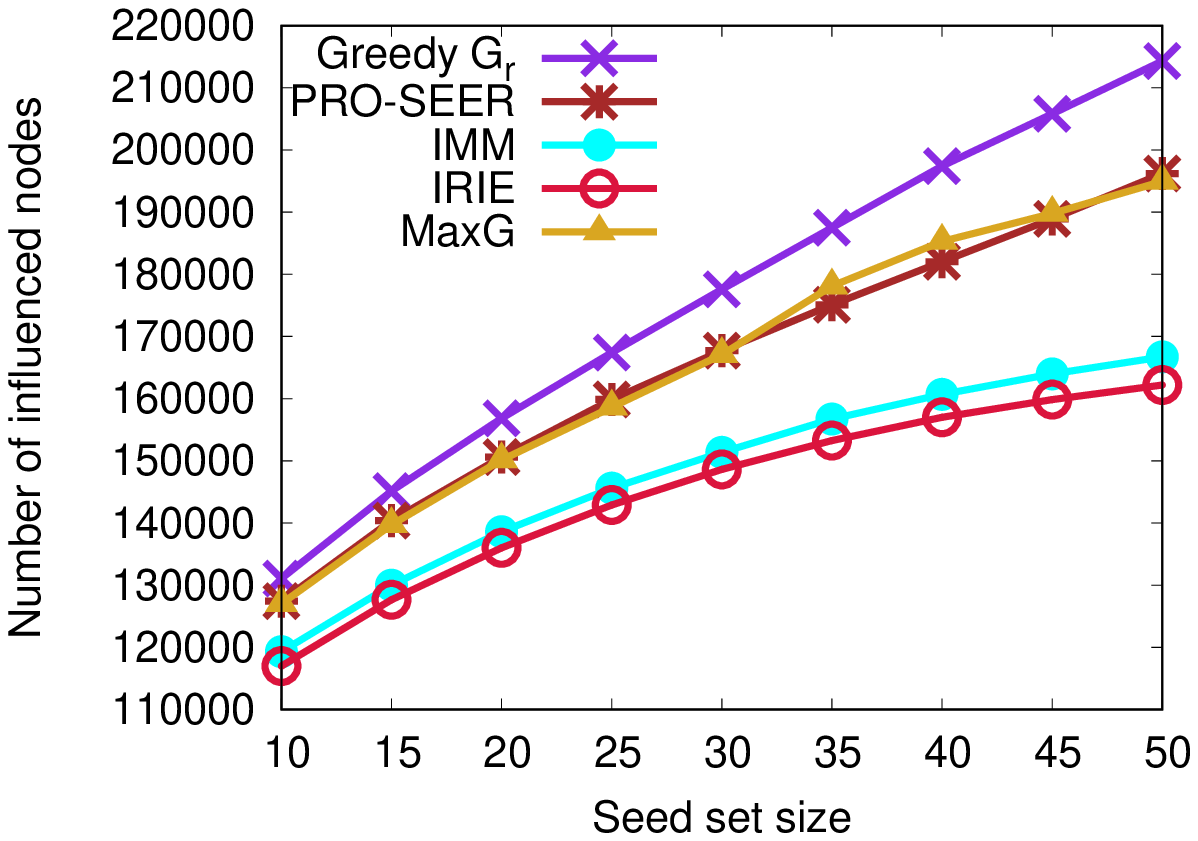, width=0.32\linewidth}}
\subfloat[][$G_0=$\textsf{Ph-$G_1$}, $G_r=$\textsf{Ph-$G_2$} (Moderate DoC)\label{f10h}]{\epsfig{file=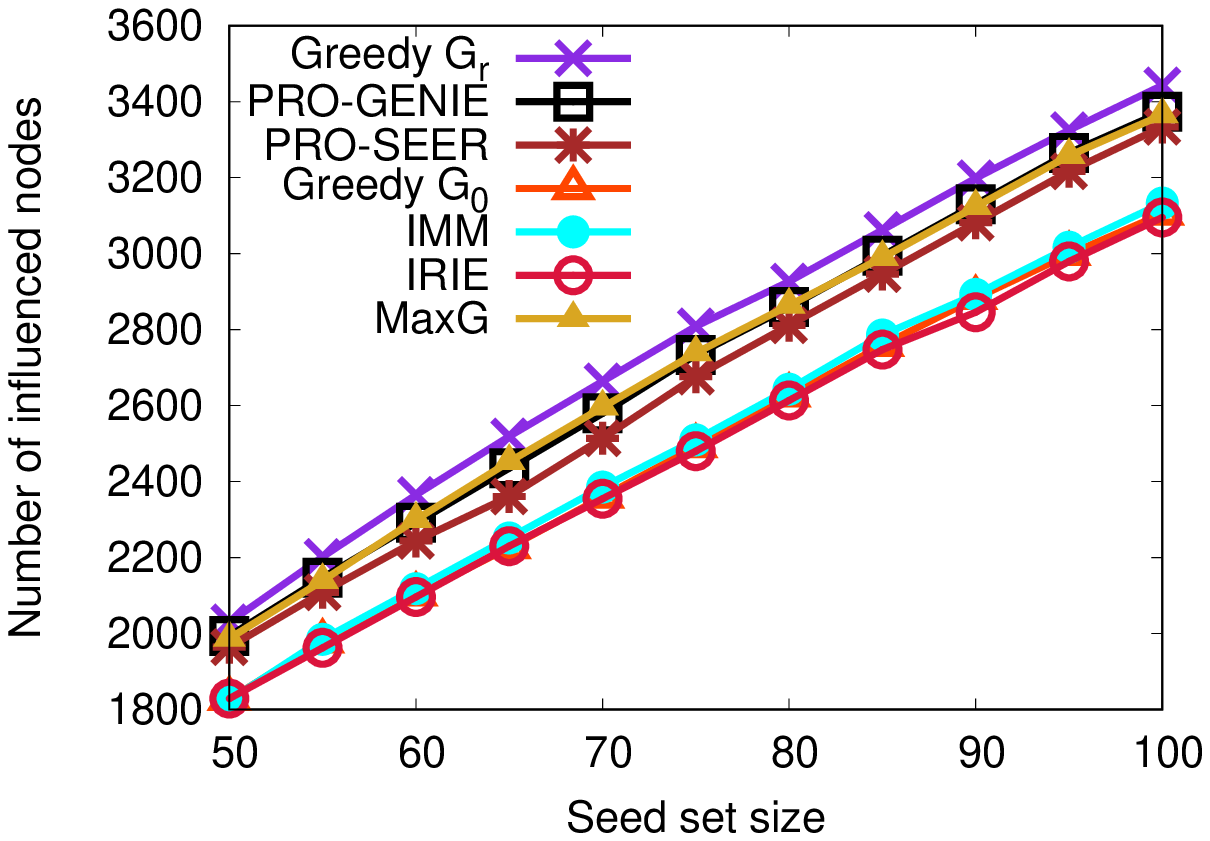, width=0.32\linewidth}}
\vspace{0ex}\caption{Influence spread vs $p$: (a-c) $p=0.05$; (d-f) $p=0.1.$}\label{fig10}
\hrule
\vspace{0ex}\end{figure*}

\textbf{\underline{Effect of $p$.}} Lastly, we test the influence spread quality of networks with different degree of change by varying $p$ to $0.05$ and $0.1$, respectively. Fig.~\ref{fig10} reports the influence spread quality for different experimental settings. Observe that the influence spread quality of different algorithms are qualitatively similar to those reported earlier for $p=0.01$.

%=============================================================================================
\vspace{0ex}\section{Conclusions \& Future Work}\label{sec:conclu}
%=============================================================================================
The classical influence maximization (\textsc{im}) problem as introduced in~\cite{Kempe2003} has been extensively studied in the literature. It has been a solid foundation for many subsequent algorithmic improvements to the \textsc{im} problem with non-trivial performance guarantees as well as novel variations to the original problem. Unfortunately, due to the inherent limitation of the problem definition, existing techniques will often generate suboptimal seeds, unless very unrealistic changes happen to the real-world social networking landscape (either social networks stop evolving during influence propagation or influence propagation takes negligible time).

In this paper, we present \textsc{proteus-im} problem, which is designed to replace the classical version by assuming that influence propagation time is not negligible, during which the underlying network evolves. Hence, it aims to find seeds in the current network that maximizes influence spread in a future instance of the network, while considering the evolution of the network during influence propagation. We propose a pair of algorithms called \textsc{proteus-genie} and \textsc{proteus-seer}, respectively, to address the \textsc{proteus-im} problem. Specifically, our algorithms address the challenge of unknown topology of the target network by predicting it using a network evolution model and then leveraging existing cascade models on this predicted network to discover seeds. Experimental results conducted over a couple of real-world datasets demonstrate that our proposed algorithms consistently outperform state-of-the-art ones designed to address the classical \textsc{im} problem.

We believe that the \textsc{proteus-im} problem can be the foundation for future research on \textsc{im} in several directions. First, it has the potential to catalyze the research community to revisit theoretical and practical aspects of existing \textsc{im} solutions (\eg algorithms, theoretical guarantees) and different variations of traditional \textsc{im} problems. Second, the state-of-the-art Forest Fire Model, which we have utilized to model network evolution, despite demonstrating superior performance in capturing network evolution, does not incorporate node or edge deletions. Hence, exploration of a more \textit{generalized} network evolution model can further enhance the influence spread quality of our proposed solutions.

%\vspace{0ex}\section*{Acknowledgment}
%
%Hui Li and Jiangtao Cui are supported by National Nature Science Foundation of China (No. 61672408, 61472298), China 111 Project (No. B16037).

\balance

\begin{thebibliography}{10}
\providecommand{\url}[1]{#1}
\csname url@samestyle\endcsname
\providecommand{\newblock}{\relax}
\providecommand{\bibinfo}[2]{#2}
\providecommand{\BIBentrySTDinterwordspacing}{\spaceskip=0pt\relax}
\providecommand{\BIBentryALTinterwordstretchfactor}{4}
\providecommand{\BIBentryALTinterwordspacing}{\spaceskip=\fontdimen2\font plus
\BIBentryALTinterwordstretchfactor\fontdimen3\font minus
  \fontdimen4\font\relax}
\providecommand{\BIBforeignlanguage}[2]{{%
\expandafter\ifx\csname l@#1\endcsname\relax
\typeout{** WARNING: IEEEtran.bst: No hyphenation pattern has been}%
\typeout{** loaded for the language `#1'. Using the pattern for}%
\typeout{** the default language instead.}%
\else
\language=\csname l@#1\endcsname
\fi
#2}}
\providecommand{\BIBdecl}{\relax}
\BIBdecl

\bibitem{Kempe2003}
D.~Kempe, J.~M. Kleinberg, and {\'E}.~Tardos, ``Maximizing the spread of
  influence through a social network,'' in \emph{KDD}.\hskip 1em plus 0.5em
  minus 0.4em\relax ACM Press, 2003, pp. 137--146.

\bibitem{Shirazipourazad:2012:IPA:2396761.2396837}
S.~Shirazipourazad, B.~Bogard, H.~Vachhani, A.~Sen, and P.~Horn, ``Influence
  propagation in adversarial setting: how to defeat competition with least
  amount of investment,'' in \emph{CIKM}.\hskip 1em plus 0.5em minus
  0.4em\relax ACM Press, 2012, pp. 585--594.

\bibitem{effMax_KDD09_Chen}
W.~Chen, Y.~Wang, and S.~Yang, ``Efficient influence maximization in social
  networks,'' in \emph{KDD}, 2009.

\bibitem{GLL11}
A.~Goyal, W.~Lu, and L.~V.~S. Lakshmanan, ``Simpath: An efficient algorithm for
  influence maximization under the linear threshold model,'' in \emph{ICDM},
  2011.

\bibitem{DBLP:conf/sigmod/NguyenTD16}
H.~T. Nguyen, M.~T. Thai, and T.~N. Dinh, ``Stop-and-stare: Optimal sampling
  algorithms for viral marketing in billion-scale networks,'' in \emph{SIGMOD},
  2016, pp. 695--710.

\bibitem{DBLP:journals/pvldb/HuangWBXL17}
K.~Huang, S.~Wang, G.~S. Bevilacqua, X.~Xiao, and L.~V.~S. Lakshmanan,
  ``Revisiting the stop-and-stare algorithms for influence maximization,''
  \emph{{PVLDB}}, vol.~10, no.~9, pp. 913--924, 2017.

\bibitem{DBLP:journals/pvldb/ChenFLFTT15}
S.~Chen, J.~Fan, G.~Li, J.~Feng, K.~Tan, and J.~Tang, ``Online topic-aware
  influence maximization,'' \emph{{PVLDB}}, vol.~8, no.~6, 2015.

\bibitem{Li:2013:CCG:2452376.2452415}
H.~Li, S.~S. Bhowmick, and A.~Sun, ``Cinema: conformity-aware greedy algorithm
  for influence maximization in online social networks,'' in \emph{EDBT}.\hskip
  1em plus 0.5em minus 0.4em\relax ACM Press, 2013, pp. 323--334.

\bibitem{Carnes:2007:MIC:1282100.1282167}
T.~Carnes, C.~Nagarajan, S.~M. Wild, and A.~van Zuylen, ``Maximizing influence
  in a competitive social network: a follower's perspective,'' in \emph{ICEC},
  2007.

\bibitem{DBLP:conf/sigmod/AroraGR17}
A.~Arora, S.~Galhotra, and S.~Ranu, ``Debunking the myths of influence
  maximization: An in-depth benchmarking study,'' in \emph{SIGMOD}, 2017, pp.
  651--666.

\bibitem{PhysRevLett.103.038702}
J.~L. Iribarren and E.~Moro, ``Impact of human activity patterns on the
  dynamics of information diffusion,'' \emph{Phys. Rev. Lett.}, vol. 103, p.
  038702, Jul 2009.

\bibitem{facebook1b}
G.~A. Fowler, ``Facebook: One billion and counting,'' \emph{The Wall Street
  Journal (Dow Jones)}, vol. 2012-10-04, 2012.

\bibitem{twitter2m}
S.~Fiegerman. (2012) Twitter now has more than 200 million monthly active
  users. http://mashable.com/2012/12/18/twitter-200-million-active-users/.

\bibitem{DBLP:conf/kdd/LeskovecBKT08}
J.~Leskovec, L.~Backstrom, R.~Kumar, and A.~Tomkins, ``Microscopic evolution of
  social networks,'' in \emph{KDD}, 2008.

\bibitem{DBLP:conf/icdm/ZhuangSTZS13}
H.~Zhuang, Y.~Sun, J.~Tang, J.~Zhang, and X.~Sun, ``Influence maximization in
  dynamic social networks,'' in \emph{ICDM}, 2013.

\bibitem{DBLP:conf/sdm/ChenSHX15}
X.~Chen, G.~Song, X.~He, and K.~Xie, ``On influential nodes tracking in dynamic
  social networks,'' in \emph{SDM}, 2015.

\bibitem{DBLP:journals/pvldb/OhsakaAYK16}
N.~Ohsaka, T.~Akiba, Y.~Yoshida, and K.~Kawarabayashi, ``Dynamic influence
  analysis in evolving networks,'' \emph{{PVLDB}}, vol.~9, no.~12, pp.
  1077--1088, 2016.

\bibitem{DBLP:journals/pvldb/WangFLT17}
Y.~Wang, Q.~Fan, Y.~Li, and K.~Tan, ``Real-time influence maximization on
  dynamic social streams,'' \emph{{PVLDB}}, vol.~10, no.~7, pp. 805--816, 2017.

\bibitem{DBLP:journals/ton/TongWTD17}
G.~Tong, W.~Wu, S.~Tang, and D.~Du, ``Adaptive influence maximization in
  dynamic social networks,'' \emph{{IEEE/ACM} Trans. Netw.}, vol.~25, no.~1,
  pp. 112--125, 2017.

\bibitem{DBLP:conf/cosn/GayraudPT15}
N.~T.~H. Gayraud, E.~Pitoura, and P.~Tsaparas, ``Diffusion maximization in
  evolving social networks,'' in \emph{COSN}, 2015.

\bibitem{DBLP:conf/sdm/AggarwalLY12}
C.~C. Aggarwal, S.~Lin, and P.~S. Yu, ``On influential node discovery in
  dynamic social networks,'' in \emph{SDM}, 2012.

\bibitem{DBLP:conf/cikm/MaYLK08}
H.~Ma, H.~Yang, M.~R. Lyu, and I.~King, ``Mining social networks using heat
  diffusion processes for marketing candidates selection,'' in \emph{CIKM},
  2008, pp. 233--242.

\bibitem{DBLP:conf/kdd/LeskovecKF05}
J.~Leskovec, J.~M. Kleinberg, and C.~Faloutsos, ``Graphs over time:
  densification laws, shrinking diameters and possible explanations,'' in
  \emph{KDD}, 2005, pp. 177--187.

\bibitem{DBLP:conf/sigmod/TangSX15}
Y.~Tang, Y.~Shi, and X.~Xiao, ``Influence maximization in near-linear time: {A}
  martingale approach,'' in \emph{SIGMOD}, 2015.

\bibitem{irie}
K.~Jung, W.~Heo, and W.~Chen, ``{IRIE:} scalable and robust influence
  maximization in social networks,'' in \emph{ICDM}, 2012.

\bibitem{DBLP:conf/soda/BorgsBCL14}
C.~Borgs, M.~Brautbar, J.~T. Chayes, and B.~Lucier, ``Maximizing social
  influence in nearly optimal time,'' in \emph{SODA}, 2014, pp. 946--957.

\bibitem{DBLP:journals/csur/ChakrabartiF06}
D.~Chakrabarti and C.~Faloutsos, ``Graph mining: Laws, generators, and
  algorithms,'' \emph{{ACM} Comput. Surv.}, vol.~38, no.~1, 2006.

\bibitem{DBLP:conf/kdd/BackstromHKL06}
L.~Backstrom, D.~P. Huttenlocher, J.~M. Kleinberg, and X.~Lan, ``Group
  formation in large social networks: membership, growth, and evolution,'' in
  \emph{KDD}, 2006, pp. 44--54.

\end{thebibliography}
\end{document}